\documentclass[12pt]{article}

\usepackage{wrapfig}
\usepackage{graphicx}
\usepackage[none]{hyphenat}
\usepackage[english]{babel}
\usepackage[utf8]{inputenc}
\usepackage[T1]{fontenc}

\usepackage{authblk}
\usepackage{tikz-cd}

\usepackage{caption}
\usepackage{csquotes}
\usepackage{mathtools,amssymb,amsfonts,amsthm}

\usepackage[breaklinks,unicode=true]{hyperref}
\usepackage[capitalise]{cleveref} 

\theoremstyle{plain}
\newtheorem{theorem}{Theorem}
\newtheorem*{theorem*}{Theorem}
\newtheorem{lemma}[theorem]{Lemma}
\newtheorem*{lemma*}{Lemma}
\newtheorem{proposition}[theorem]{Proposition}
\newtheorem*{proposition*}{Proposition}
\newtheorem{corollary}[theorem]{Corollary}
\newtheorem*{corollary*}{Corollary}

\theoremstyle{definition}
\newtheorem{definition}{Definition}
\newtheorem*{definition*}{Definition}
\newtheorem{example}{Example}
\newtheorem*{example*}{Example}

\crefname{theorem}{Theorem}{Theorems}
\Crefname{theorem}{Theorem}{Theorems}
\crefname{lemma}{Lemma}{Lemmas}
\Crefname{lemma}{Lemma}{Lemmas}
\crefname{proposition}{Proposition}{Propositions}
\Crefname{Prop}{Proposition}{Propositions}
\crefname{corollary}{Corollary}{Corollaries}
\Crefname{corollary}{Corollary}{Corollaries}
\crefname{definition}{Definition}{Definitions}
\Crefname{definition}{Definition}{Definitions}
\crefname{example}{Example}{Examples}
\Crefname{example}{Example}{Examples}
\crefname{theorem}{Theorem}{Theorems}

\newtheorem*{exercise*}{Exercise}
\crefname{exercise}{exercise}{exercises}
\Crefname{exercise}{Exercise}{Exercises}  

\theoremstyle{remark}
\newtheorem{remarkx}{Remark}
\newtheorem*{remarkx*}{Remark}
\crefname{remark}{Remark}{Remarks}
\Crefname{remark}{Remark}{Remarks}
\newenvironment{remark}
  {\pushQED{\qed}\remarkx}
  {\popQED\endremarkx}
\newenvironment{remark*}
  {\pushQED{\qed}\remarkx*}
  {\popQED\endremarkx*}

\title{Characteristic foliations of material evolution: from remodeling to aging}

\author[-]{V\'ictor Manuel Jim\'enez}
\author[+*]{Manuel de León}
\author[ç]{Marcelo Epstein}

\affil[1]{\href{mailto:victormanuel.jimenez@uah.es}{victormanuel.jimenez@uah.es}}

\affil[2]{\href{mailto:mdeleon@icmat.es}{mdeleon@icmat.es}}

\affil[3]{\href{mailto:epstein@enme.ucalgary.ca}{epstein@enme.ucalgary.ca}}

\affil[-]{Universidad de Alcal\'a (UAH), Departamento de F\'isica y Matem\'aticas.
Av. de Le\'on, 4A, 28805 Alcalá de Henares, Madrid, Spain} 

\affil[*]{Instituto de Ciencias Matem\'aticas (CSIC-UAM-UC3M-UCM),
C\textbackslash Nicol\'as Cabrera, 13-15, Campus Cantoblanco, UAM
28049 Madrid, Spain} 

\affil[+]{Real Academia de Ciencias Exactas, Fisicas y Naturales, C/de Valverde
22, 28004 Madrid, Spain}

\affil[ç]{Department of Mechanical Engineering. University of Calgary. 2500 University Drive NW, Calgary, Alberta, Canada, T2N IN4}

\date{\today}

\begin{document}

\sloppy

\maketitle

\thanks{}

\begin{abstract}

For any body-time manifold $\mathbb{R} \times \mathcal{B}$ there exists a groupoid, called \textit{material groupoid}, encoding all the material properties of the evolution material. A smooth distribution, the \textit{material distribution}, is constructed to deal with the case in which the material groupoid is not a Lie groupoid. This new tool provides a unified framework to deal with general non-uniform evolution materials.
\end{abstract}

\maketitle

\tableofcontents

\section{Introduction}

In this paper, we will use the approach developed by Walter Noll [6], based on the notion of the so-called material diffeomorphism between pairs of points in the body, namely, a map between the respective tangent spaces that renders the constitutive responses identical. In Noll’s terminology, a body is said to be materially uniform if all of its points are mutually materially isomorphic. In a uniform body, a smooth field of material isomorphisms is nothing but a distant parallelism whose integrability is equivalente to the homogeneity of the body. The existence of material symmetries works like a gauge freedom for these parallelims, so that the geometric notion associated to the body is the so-called $G$-structure, where $G$ is a nodel for the group of material symmetries.. In this context, local homogeneity is equivalent to the integrability of the associated $G$-structure. However, this approach depends on the choice of a linear frame (some archetype) and it cannot be used for non-uniform bodies.

So, the natural extension of the theory was the consideration of a more general kind of algebraic/geometric structures, say Lie groupoids. Indeed, the point now is to consider all the material isomorphisms at the same time, which provides just the so-called material groupoid. This theory has been extensively developed in the book \cite{VMMDME} (see also \cite{CHARDIST,MGEOEPS,MD}) even for more general materials where the material groupoid is not a Lie groupoid.

One of the contributions of the paper is the construction of groupoids canonically associated with the evolution of a material \cite{EPSTEIN201572,MEPMDLSEG,EPSBOOK2}, to present a new framework. This has been done to study simple materials in previous papers, but we now consider this new scenario. This technique will permit us to define a type of global remodeling of non-uniform bodies (Definition \ref{definition17}) and a definition of differentiable global aging (Definition \ref{smoothaging342345}).

Recall that the infinitesimal approximation of a Lie groupoid is its Lie algebroid, just as the tangent space to a manifold at a point is the linear approximation of a neighborhood of the point. One of the constructions we have developed in several previous papers is that of the \textit{characteristic distribution} associated to a subgroupoid of a Lie groupoid, even when the subgroupoid is not differentiable. This construction is somehow a generalization of the Lie algebroid associated to a Lie groupoid \cite{CHARDIST,VMMDME}.

The characteristic distributions associated to material groupoids are constructed, which give rise to the respective foliations. Uniform aging is also presented for the first time in this paper. It expresses that, although the body ages, all material points age in the same way over time.

Instead to consider a simple body $\mathcal{B}$ as in simple materials, we consider the so-called body-time manifold as the fibre bundle $\mathcal{C} = \mathbb{R} \times \mathcal{B}$ over $\mathbb{R}$. The embbeddings are now fiber bundle embbeddings into the trivial fibre bundle $\mathbb{R}\times \mathbb{R}^{3}$ over $\mathbb{R}$; such an embbedding $\Phi$ is called a \textit{history}, since for a given pòint $X_{0}$ in $\mathcal{B}$, the family $\Phi_{t}\left( X_{0}\right)$ describes the history of the material point. The crucial point now is to consider the vertical subbundle associated to the body-time manifold $\mathcal{C}$, $\mathcal{V}$, and the associated frame groupoid $\Phi \left( \mathcal{V} \right) \rightrightarrows \mathcal{C}$. This groupoid will play a similar role of the 1-jets groupoid $\Pi^{1}\left( \mathcal{B}, \mathcal{B}\right)$ on $\mathcal{B}$ for elastic simple materials (see Part \ref{partelasticmat}). This permits to define the corresponding material groupoid $\Omega \left( \mathcal{C}\right)$ for a given constitutive law as a subgroupoid of $\Phi \left( \mathcal{V} \right)$, consisting just in those material isomorphisms, connecting not only material points but also the different instants of time. So, we may consider a temporal counterpart of uniformity called remodeling. Indeed, a material particle $X \in \mathcal{B}$ is presenting a remodeling when it is connected with all the instants by a material isomorphism, i.e., all the points at $\mathbb{R} \times \{X\}$ are connected by material isomorphisms; if this happen for all the material points, then it is said that $\mathcal{C}$ presents a global remodeling. In other words, the material properties of the body do not change along the time. Particular cases are the phenomena of growth and resorption (remodeling with volume increase or volume decrease of the material body). This kind evolution may be found in biological tissues \cite{RODRIGUEZ1994455} or Wolff’s law of trabecular architecture of bones \cite{TURNER19921}).

We also consider the aging phenomenum. The definition is very simple; a material particle $X \in \mathcal{B}$ is presenting a aging when it is not presenting a remodeling, i.e., not all the instants are connected by a material isomorphism. In other words, the material response is not preserved along the time via material isomorphisms, and the constitutive properties are changing with the time. Our approach allow us to introduce the concept of smooth aging.

In a more technical way, these are some of the most relevant results contained in the paper:\\

\noindent{\textbf{Corollary \ref{4.4.second2324.uniform2445}:}\\
Let be a body-time manifold $\mathcal{C}$ with some (and hence all of them) state uniform. $\mathcal{C}$ is presenting a smooth uniform remodeling if, and only if, $\Omega \left( \mathcal{C}\right)$ is a transitive Lie subgroupoid of $\Phi \left( \mathcal{V} \right)$.}\\

\noindent{\textbf{Theorem \ref{14.1.second323.globalevolutionsd}:}\\
Let be a body-time manifold $\mathcal{C}$. The body-material foliation $\mathcal{F}$ (resp. uniform material foliation $\mathcal{G}$) divides $\mathcal{C}$ into maximal smooth uniform remodeling processes (resp. uniform remodeling processes).}\\

\noindent{\textbf{Theorem \ref{14.1.second323.globalevolutionsd.dimesions}:}\\
Let be a body-time manifold $\mathcal{C}$. $\mathcal{C}$ presents a smooth uniform remodeling process (resp. uniform remodeling) if, and only if, $dim  \left( A \Omega \left( \mathcal{C} \right)^{\sharp}_{\left(t,X\right)}\right) = 4 $ (resp. $dim  \left( A \Omega \left( \mathcal{C} \right)^{B}_{\left(t,X\right)}\right) = 4 $) for all instant $t$ and particle $X$, with $A \Omega \left( \mathcal{C}\right)^{\sharp}_{\left(t,X\right)}$ (resp. $A \Omega \left( \mathcal{C}\right)^{B}_{\left(t,X\right)}$) the fibre of $A \Omega \left( \mathcal{C}\right)^{\sharp}$ (resp. $A \Omega \left( \mathcal{C}\right)^{B}$) at $\left(t,X\right)$.\\
}

Roughly speaking, \textbf{Corollary \ref{4.4.second2324.uniform2445}} establishes that the uniform differentiable remodeling is equivalent to $\Omega \left( \mathcal{C}\right)$ being a transitive Lie subgrupoid of $\Phi \left( \mathcal{V} \right)$. This clarifies the difference between remodeling and uniform differentiable remodeling. \textbf{Theorem \ref{14.1.second323.globalevolutionsd}} proves that there are two maximal foliations of $\mathcal{C}$ separating the evolution of the material into uniform remodeling and smooth uniform remodelings, respectively. On the other hand, \textbf{Theorem \ref{14.1.second323.globalevolutionsd.dimesions}} shows that, using the material distribution, one can imagine the shape of the foliation associated with the differentiable uniform remodeling by calculating the dimensions of its leaves. In particular, if the dimension is 4 at any material point, the evolution has a uniform and differentiable remodeling. Thus, studying whether the evolution has a uniform and differentiable remodeling is reduced to the study of the linear equation (\ref{Eqmaterialgroupoid12timedependent}).

On the other hand, \textbf{Theorem \ref{14.1.second}} determines the material foliation by uniformly differentiable components of the body at each of the instants. As one can imagine, in \textbf{Proposition \ref{proprelacionimport234}} it is proved that if one freezes the evolution in the leaves given by \textbf{Theorem \ref{14.1.second323.globalevolutionsd}}, the leaves of \textbf{Theorem \ref{14.1.second}} are recovered. Finally, \textbf{Theorem  \ref{computationalproposition24124}} is the analogue of \textbf{Theorem \ref{14.1.second323.globalevolutionsd.dimesions}} for differentiable remodeling, thus giving a \textit{computational condition} (linear equation) for studying differentiable remodeling.

Due to the length of this paper, it is important to note that Parts I and II consist mainly of preliminaries to make the text as self-contained as possible. Indeed, Part I introduces the basic concepts of groupoids, following the references \cite{KMG,JNM}. It also includes a construction that is essential throughout the paper, the so-called \textit{characteristic distribution}, which was introduced for the first time in \cite{CHARDIST}. Part II is devoted to a quick introduction to the theory of simple materials and the concept of uniformity. A first use of the characteristic distribution is shown in this part following \cite{MD,MGEOEPS,GENHOM} (see also the book \cite{VMMDME}). Finally, section 4 in part III is devoted to the introduction of the concept of material evolution in a very abstract and setting. A reader who is familiar with these topics could skip these two parts and go directly to the three last parts of the paper where the new results are  described.

\part{Groupoids and distributions}

\section{Groupoids}
We will start with a very brief introduction on \textit{(Lie) groupoids} which turns out to be crucial to understand the results shown in this paper. Groupoids are a natural generalization of groups which were presented in 1926 by Brandt \cite{HBU}. Furthermore, adding differential structures we obtain the notion of \textit{Lie groupoid} which was firtly introduced by Ehresmann in a series of articles \cite{CELC,CELP,CES,CEC} and redefined in \cite{JPRA} by Pradines.\\
We will follow the most relevant reference on groupoids \cite{KMG}. In \cite{EPSBOOK} and \cite{WEINSGROUP} we can find a more intuitive view of this topic. The book \cite{JNM} (in Spanish) is also recommendable as a rigurous introduction to groupoids.
\begin{definition}
\rm
Let $ M$ be a set. A \textit{groupoid} over $M$ is given by a set $\Gamma$ provided with the maps $\alpha,\beta : \Gamma \rightarrow M$ (\textit{source map} and \textit{target map} respectively), $\epsilon: M \rightarrow \Gamma$ (\textit{section of identities}), $i: \Gamma \rightarrow \Gamma$ (\textit{inversion map}) and $\cdot : \Gamma_{\left(2\right)} \rightarrow \Gamma$ (\textit{composition law}) where for each $k \in \mathbb{N}$, $\Gamma_{\left(k\right)}$ is given by $k$ points $ \left(g_{1}, \hdots , g_{k}\right) \in \Gamma \times \stackrel{k)}{\ldots} \times \Gamma $ such that $\alpha\left(g_{i}\right)=\beta\left(g_{i+1}\right)$ for $i=1, \hdots , k -1$. It satisfy the following properties:\\
\begin{itemize}
\item[(1)] $\alpha$ and $\beta$ are surjective and for each $\left(g,h\right) \in \Gamma_{\left(2\right)}$,
$$ \alpha\left(g \cdot h \right)= \alpha\left(h\right), \ \ \ \beta\left(g \cdot h \right) = \beta\left(g\right).$$
\item[(2)] Associative law with the composition law, i.e.,
$$ g \cdot \left(h \cdot k\right) = \left(g \cdot h \right) \cdot k, \ \forall \left(g,h,k\right) \in \Gamma_{\left(3\right)}.$$
\item[(3)] For all $ g \in \Gamma$,
$$ g \cdot \epsilon \left( \alpha\left(g\right)\right) = g = \epsilon \left(\beta \left(g\right)\right)\cdot g .$$
Therefore,
$$ \alpha \circ  \epsilon \circ \alpha = \alpha , \ \ \ \beta \circ \epsilon \circ \beta = \beta.$$
Since $\alpha$ and $\beta$ are surjetive, we have that
$$ \alpha \circ \epsilon = Id_{M}, \ \ \ \beta \circ \epsilon = Id_{M},$$
where the map $Id_{M}$ is the identity map at $M$.
\item[(4)] For each $g \in \Gamma$,
$$i\left(g\right) \cdot g = \epsilon \left(\alpha\left(g\right)\right) , \ \ \ g \cdot i\left(g\right) = \epsilon \left(\beta\left(g\right)\right).$$
Then,
$$ \alpha \circ i = \beta , \ \ \ \beta \circ i = \alpha.$$
\end{itemize}
These maps are called \textit{structure maps}. The usual notation for a grupoid is $ \Gamma \rightrightarrows M$.
\end{definition}
\noindent{$M$ is denoted by $\Gamma_{\left(0\right)}$ and it is identified with the set $\epsilon \left(M\right)$ of identities of $\Gamma$. $\Gamma$ is also denoted by $\Gamma_{\left(1\right)}$. The elements of $M$ are called \textit{objects} and the elements of $\Gamma$ are called \textit{morpishms}. Furthermore, for each $g \in \Gamma$ the element $i \left( g \right)$ is called \textit{inverse of $g$} and it is denoted by $g^{-1}$.}\\

\begin{definition}
\rm
Let $\Gamma \rightrightarrows M$ be a groupoid. The map $\left(\alpha , \beta\right) : \Gamma \rightarrow M \times M$ is called the \textit{anchor map}. The space of sections of the anchor map is denoted by $\Gamma_{\left(\alpha, \beta\right)} \left(\Gamma\right)$.
\end{definition}  
Roughly speaking, a groupoid may be thought as a set of ``\textit{arrows}'' ($\Gamma$) joining points ($M$) next to a composition law with similar rules to the composition of maps.
\begin{definition}
\rm

If $\Gamma_{1} \rightrightarrows M_{1}$ and $\Gamma_{2} \rightrightarrows M_{2}$ are two groupoids then a \textit{morphism of groupoids} from $\Gamma_{1} \rightrightarrows M_{1}$ to $\Gamma_{2} \rightrightarrows M_{2}$ consists of two maps $\Phi : \Gamma_{1} \rightarrow \Gamma_{2}$ and $\phi : M_{1} \rightarrow M_{2}$ satisfying the commutative relations of the following diagrams, 

\begin{center}
 \begin{tikzcd}[column sep=huge,row sep=huge]
\Gamma_{1}\arrow[r, "\Phi"] &\Gamma_{2} \arrow[d, "\alpha_{2}"] &   &    \Gamma_{1}\arrow[r, "\Phi"] &\Gamma_{2} \arrow[d, "\beta_{2}"]\\
 M_{1} \arrow[u,"\alpha_{1}"] \arrow[r,"\phi"] & M_{2} &  &  M_{1} \arrow[u,"\beta_{1}"] \arrow[r,"\phi"] & M_{2}
 \end{tikzcd}
\end{center}

\begin{center}
 \begin{tikzcd}[column sep=huge,row sep=huge]
\left(\Gamma_{1}\right)_{(2)} \arrow[d,"\Phi_{(2)}"] \arrow[r, "\cdot"] &\Gamma_{1} \arrow[d, "\Phi"] \\
\left(\Gamma_{2}\right)_{(2)} \arrow[r,"\cdot"] & \Gamma_{2} 
 \end{tikzcd}
\end{center}
where, 
$$\Phi_{(2)} \left( g_{1} , h_{1} \right) = \left( \Phi \left(g_{1} \right) ,  \Phi \left(h_{1} \right) \right)$$
for all $\left( g_{1} , h_{1} \right) \in \left(\Gamma_{1}\right)_{(2)}$. Equivalently, for any $g_{1} \in \Gamma_{1}$
\begin{equation}\label{4}
\alpha_{2} \left( \Phi \left(g_{1}\right)\right) = \phi \left(\alpha_{1} \left(g_{1} \right)\right), \ \ \ \ \ \ \ \beta_{2} \left( \Phi \left(g_{1}\right)\right) = \phi \left(\beta_{1} \left(g_{1} \right)\right),
\end{equation}
where $\alpha_{i}$ and $\beta_{i}$ are the source and the target maps of $\Gamma_{i} \rightrightarrows M_{i}$ respectively, for $i=1,2$, and preserves the composition, i.e.,
$$\Phi \left( g_{1} \cdot h_{1} \right) = \Phi \left(g_{1}\right) \cdot \Phi \left(h_{1}\right), \ \forall \left(g_{1} , h_{1} \right) \in \Gamma_{\left(2\right)}.$$
We will denote this morphism as $\Phi$.
\end{definition}

\noindent{An immediate consequence is that $\Phi$ preserves the identities, i.e.,}
$$\Phi \circ  \epsilon_{1} = \epsilon_{2} \circ \phi,$$
where $\epsilon_{i}$ is the section of identities of $\Gamma_{i} \rightrightarrows M_{i}$ for $i=1,2$.\\
Using the notion of morphism of groupoids, we may define a \textit{subgroupoid} of a groupoid $\Gamma \rightrightarrows M$ as a groupoid $\Gamma' \rightrightarrows M'$ such that $M' \subseteq M$, $\Gamma' \subseteq \Gamma$ and the inclusion map is a morphism of groupoids. More explicitly, $\Gamma' \left( \subseteq \Gamma\right) \rightrightarrows M' \left(\subseteq M\right)$ is a subgroupoid of $\Gamma \rightrightarrows M$ if it is groupoid with the same structure maps than $\Gamma$.
\begin{example}\label{5}
\rm
A group is a groupoid over a point. Indeed, let $G$ be a group and $e$ the identity element of $G$. Then, $G \rightrightarrows \{e\}$ is a groupoid, where the operation law of the groupoid, $\cdot$, is the operation in $G$.
\end{example}
\begin{example}\label{7}

\rm
For any set $A$, we shall consider the product space $ A \times A$. Then, the maps,
\begin{itemize}
\item[] $\alpha \left(a,b\right) = a , \ \ \beta \left(a,b\right)=b, \ \forall \left(a,b\right) \in  A \times A$
\item[] $\left(c,b\right)\cdot\left(a,c\right)= \left(a,b\right), \ \forall \left( c,b\right),\left(a,c\right) \in  A \times A$
\item[] $ \epsilon \left(a\right) = \left(a,a\right), \ \forall a \in A$
\item[] $ \left(a,b\right)^{-1}=\left(b,a\right), \ \forall \left(a,b\right) \in  A \times A$
\end{itemize}
endow $A \times A$ with a structure of groupoid over $A$, called the \textit{pair groupoid}.
\end{example}

\noindent{Observe that, if $\Gamma \rightrightarrows M$ is an arbitrary groupoid over $M$, then the anchor map $\left(\alpha , \beta\right) : \Gamma \rightarrow M \times M$ is a morphism from $\Gamma \rightrightarrows M$ to the pair groupoid of $M$.\\
Next, let us give the key example of groupoid in this paper.}
\begin{example}\label{8}
\rm
Let us consider a vector bundle $A$ on a manifold $M$. For each $z\in M$, denote by $A_{z}$ the fibre of $A$ over $z$. Then, $\Phi \left(A\right)$ is the set of linear isomorphisms $L_{x,y}: A_{x} \rightarrow A_{y}$, for $x,y \in M$ and it may be endowed with the structure of groupoid with the following structure maps,
\begin{itemize}
\item[(i)] $\alpha\left(L_{x,y}\right) = x$
\item[(ii)] $\beta\left(L_{x,y}\right) = y$
\item[(iii)] $L_{y,z} \cdot G_{x,y} = L_{y,z} \circ G_{x,y}, \ L_{y,z}: A_{y} \rightarrow A_{z}, \ G_{x,y}: A_{x} \rightarrow A_{y}$
\end{itemize}
This groupoid is called the \textit{frame groupoid on $A$}. A particular relevant case arises when we choose $A$ equal to the tangent bundle $TM$ of $M$. In this latter case, the groupoid will be called \textit{1-jets groupoid on} $M$ and denoted by $\Pi^{1} \left(M,M\right)$. Notice that any isomorphism $L_{x,y}: T_{x}M \rightarrow T_{y}M$ may be written as a $1-$jet $j_{x,y}^{1} \psi$ of a local diffeomorphism $\psi$ from $M$ to $M$ such that $\psi \left( x \right) = y$. Remember that the $1-$jet $j_{x,y}^{1} \psi$ is given by that induced tangent map $T_{x}\psi: T_{x}M \rightarrow T_{y}M$. To study in detail the formalism of $1-$jets see \cite{SAUND}.
\end{example}
\begin{definition}\label{58}
\rm
Let $\Gamma \rightrightarrows M$ be a groupoid with $\alpha$ and $\beta$ the source map and target map, respectively. For each $x \in M$, the set
$$\Gamma^{x}_{x}= \beta^{-1}\left(x\right) \cap \alpha^{-1}\left(x\right),$$
is called the \textit{isotropy group of} $\Gamma$ at $x$. The set
$$\mathcal{O}\left(x\right) = \beta\left(\alpha^{-1}\left(x\right)\right) = \alpha\left(\beta^{-1}\left( x\right)\right),$$
is called the \textit{orbit} of $x$, or \textit{the orbit} of $\Gamma$ through $x$.
\end{definition}

\noindent{Notice that the orbit of a point $x$ consists of the points which are ``\textit{connected}" with $x$ by a morphism in the groupoid while the isotropy group is given by the morphisms connecting $x$ with $x$. Of course, the composition law is globally defined inside the isotropy groups. Thus, the isotropy groups inherits a \textit{bona fide} group structure.}
\begin{definition}
\rm
If $\mathcal{O}\left(x\right) =\{x\}$, or equivalently $\beta^{-1}\left(x\right) = \alpha^{-1}\left( x \right)=\Gamma_{x}^{x}$, then $x$ is called a \textit{fixed point}.  \textit{The orbit space of} $\Gamma$ is the space of orbits of $\Gamma $ on $M$. If $\mathcal{O}\left(x\right) = M$ for all $x \in M$ (or equivalently $\left(\alpha,\beta\right) : \Gamma  \rightarrow M \times M$ is a surjective map) the groupoid $\Gamma \rightrightarrows M$ is called \textit{transitive}. If every $x \in M$ is fixed point, then the groupoid $\Gamma \rightrightarrows M$ is called \textit{totally intransitive}. Furthermore, a subset $N$ of $M$ is called \textit{invariant} \index{Invariant} if it is a union of some orbits.\\
Finally, the sets,
$$ \alpha^{-1} \left(x \right) = \Gamma_{x}, \ \ \ \ \ \beta^{-1} \left(x \right) = \Gamma^{x},$$
are called $\alpha-$\textit{fibre at} $x$ and $\beta-$\textit{fibre at $x$}, respectively.
\end{definition}

\begin{definition}\label{9}
\rm
Let $\Gamma \rightrightarrows M$ be a groupoid. We may define the \textit{left translation on $g \in \Gamma$} as the map $L_{g} : \Gamma^{\alpha\left(g\right)} \rightarrow \Gamma^{\beta\left(g\right)}$, given by
$$ h \mapsto  g \cdot h .$$
We may define the right translation on $g$, $R_{g} : \Gamma_{\beta\left(g\right)} \rightarrow \Gamma_{ \alpha \left(g\right)}$ analogously. 
\end{definition}
\noindent{Note that, the identity map on $\Gamma^{x}$ may be written as the following translation map,}
\begin{equation}\label{10} 
Id_{\Gamma^{x}} = L_{\epsilon \left(x\right)}.
\end{equation}
For any $ g \in \Gamma $, the left (resp. right) translation on $g$, $L_{g}$ (resp. $R_{g}$), is a bijective map with inverse $L_{g^{-1}}$ (resp. $R_{g^{-1}}$).\\

Different kind of structures may be imposed on a groupoid. In particular, we are interested in the so-called \textit{Lie groupoids} which are endowed with a differentiable structure.
\begin{definition}
\rm
A \textit{Lie groupoid} is a groupoid $\Gamma \rightrightarrows M$ such that $\Gamma$ is a smooth manifold, $M$ is a smooth manifold and the structure maps are smooth. Furthermore, the source and the target map are submersions.\\
A \textit{Lie groupoid morphism} is a groupoid morphism which is differentiable. An \textit{embedding of Lie groupoids} is a Lie groupoid morphism $\left( \Phi , \phi \right)$ such that $\Phi$ and $\phi$ are embeddings. A \textit{Lie subgroupoid} of $\Gamma \rightrightarrows M$ is a Lie groupoid $\Gamma' \rightrightarrows M'$ such that $\Gamma' $ and $M'$ are submanifolds of $\Gamma$ and $M$ respectively, and the inclusion maps $i_{\Gamma'} : \Gamma' \hookrightarrow \Gamma$ $i_{M'} : M' \hookrightarrow M$ become a morphism of Lie groupoids. $\Gamma' \rightrightarrows M'$ is said to be a \textit{reduced Lie subgroupoid} if it is transitive and $M'=M$.
\end{definition}
It is not difficult to check that if there exists a reduced Lie subgroupoid of a groupoid $\Gamma \rightrightarrows M$, then $\Gamma \rightrightarrows M$ is transitive.\\\\
\noindent{Notice that the following statements are immediate:}
\begin{itemize}
    \item $\epsilon$ is an injective immersion.
    \item For each $g \in \Gamma$, the left translation $L_{g}$ (resp. right translation $R_{g}$) is a diffeomorphism, for all $g \in \Gamma$.
    \item For each $k \in \mathbb{N}$, $\Gamma_{\left(k\right)}$ is a smooth manifold, for all $k \in \mathbb{N}$.
    \item The $\beta-$fibres and the $\alpha-$fibres are closed submanifolds of $\Gamma$.
\end{itemize}
\begin{example}\label{11}
\rm
A Lie group is a Lie groupoid over a point. 
\end{example}
\begin{example}\label{45}
\rm
Let $M$ be a manifold. The pair groupoid $M \times M \rightrightarrows M$ is a Lie groupoid.
\end{example}

\begin{example}\label{15}
\rm
The frame groupoid $\Phi \left( A \right)$ on a vector bundle $A$ is a Lie groupoid (see example \ref{8}). Indeed, let $\left(x^{i}\right)$ and $\left(y^{j}\right)$ be local coordinates on open neighbourhood $U, V \subseteq M$ and $\{\alpha_{p}\}$ and $\{ \beta_{q}\}$ be local basis of sections of $A_{U}$ and $A_{V}$ respectively. The corresponding local coordinates $\left(x^{i} \circ \pi, \alpha^{p}\right)$ and $\left(y^{j} \circ \pi, \beta^{q}\right)$ on $A_{U}$ and $A_{V}$ are given by
\begin{itemize}
\item For all $a \in A_{U}$,
$$ a =  \alpha^{p}\left(a\right) \alpha_{p}\left(x^{i}\left(\pi \left(a\right)\right)\right).$$\\
\item For all $a \in A_{V}$,
$$ a =  \beta^{q}\left(a\right) \beta_{q}\left(y^{j}\left(\pi \left(a\right)\right)\right).$$
\end{itemize}
\noindent{Then, it may be constructed a local coordinate system on $\Phi \left(A\right)$
$$ \Phi \left(A_{U,V}\right) : \left(x^{i} , y^{j}_{i}, y^{j}_{i}\right),$$
where, $A_{U,V} = \alpha^{-1}\left(U\right) \cap \beta^{-1}\left(V\right)$ and for each $L_{x,y} \in \alpha^{-1}\left(x\right) \cap \beta^{-1}\left(y\right) \subseteq \alpha^{-1}\left(U\right) \cap \beta^{-1}\left(V\right)$,}
\begin{itemize}\label{16}
\item $x^{i} \left(L_{x,y}\right) = x^{i} \left(x\right)$.
\item $y^{j} \left(L_{x,y}\right) = y^{j} \left( y\right)$.
\item $y^{j}_{i}\left( L_{x,y}\right)  = A_{L_{x,y}}$, where $A_{L_{x,y}}$ is the associated matrix to the induced map of $L_{x,y}$ by the local coordinates $\left(x^{i} \circ \pi, \alpha^{p}\right)$ and $\left(y^{j} \circ \pi, \beta^{q}\right)$.
\end{itemize}
In the particular case of the $1-$jets groupoid on $M$, $\Pi^{1}\left(M,M\right)$, the local coordinates will be denoted as follows
\begin{equation}\label{17}
\Pi^{1}\left(U,V\right) : \left(x^{i} , y^{j}, y^{j}_{i}\right),
\end{equation}
where, for each $ j^{1}_{x,y} \psi \in \Pi^{1}\left(U,V\right)$
\begin{itemize}
\item $x^{i} \left(j^{1}_{x,y} \psi\right) = x^{i} \left(x\right)$.
\item $y^{j} \left(j^{1}_{x,y}\psi \right) = y^{j} \left( y\right)$.
\item $y^{j}_{i}\left( j^{1}_{x,y}\psi\right)  = \dfrac{\partial \left(y^{j}\circ \psi\right)}{\partial x^{i}_{| x} }$.
\end{itemize}

\end{example}
The most important example of groupoid in this paper will be the \textit{material groupoid} which will be constructed as a subgroupoid of special cases of the frame groupoid. In particular, we will deal with the $1-$jets groupoid $\Pi^{1}\left(\mathcal{B}, \mathcal{B}\right)$ on a manifold $\mathcal{B}$ (body) and a frame groupoid $\Phi \left( \mathcal{V} \right)$ of the vertical bundle $\mathcal{V}$ of a given vector bundle $\mathcal{C}$ (material evolution).

\section{Characterististic distribution}

\noindent{From now on, we will consider the following elements:  \textit{$ \Gamma \rightrightarrows M$ will be a Lie groupoid and $\overline{\Gamma}$ will be a subgroupoid of $\Gamma$ (not necessarily a Lie subgroupoid of $\Gamma$) over the same manifold $M$}.\\\\
\noindent{We will also denote by $\overline{\alpha}$, $\overline{\beta}$, $\overline{\epsilon}$ and $\overline{i}$ the restrictions of the structure maps $\alpha$, $\beta$, $\epsilon$ and $i$ of $\Gamma$ to $\overline{\Gamma}$ (see the diagram below)}}\\

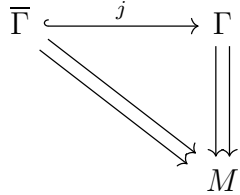
\begin{center}
 \begin{tikzcd}[column sep=huge,row sep=huge]
\overline{\Gamma}\arrow[r, hook, "j"] \arrow[rd, shift right=0.5ex] \arrow[rd, shift left=0.5ex]&\Gamma \arrow[d, shift right=0.5ex] \arrow[d, shift left=0.5ex] \\
& M 
 \end{tikzcd}
\end{center}

\noindent{where $j$ is the inclusion map. Thus, we will construct the so-called \textit{characteristic distribution} $A \overline{\Gamma}^{T}$ (\cite{VMMDME,CHARDIST}).\\
A (local) vector field $\Theta \in \frak X_{loc} \left( \Gamma \right)$ on $\Gamma$ will be called \textit{admissible} for the couple $\left( \Gamma , \overline{\Gamma} \right)$ if it satisfies that,
\begin{itemize}
\item[(i)] $\Theta$ is tangent to the $\beta-$fibres, 
$$ \Theta \left( g \right) \in T_{g} \beta^{-1} \left( \beta \left( g \right) \right),$$
for all $g$ in the domain of $\Theta$.
\item[(ii)] $\Theta$ is invariant by left translations,
$$ \Theta \left( g \right) = T_{\epsilon \left( \alpha \left( g \right) \right) } L_{g} \left( \Theta \left( \epsilon \left( \alpha \left( g \right) \right) \right) \right),$$
for all $g $ in the domain of $\Theta$.
\item[(iii)] The (local) flow $\varphi^{\Theta}_{t}$ of $\Theta$ satisfies
$$\varphi^{\Theta}_{t} \left( \epsilon \left( x \right)\right) \subseteq \overline{\Gamma}, $$
for all $x \in M$.
\end{itemize}
So, roughly speaking, an admissible vector field is a left invariant vector field on $\Gamma$ whose flow at the identities is totally contained in $\overline{\Gamma}$. We denotes the family of admissible vector fields for the couple $\left( \Gamma , \overline{\Gamma} \right)$ by $\mathcal{C}_{\left( \Gamma , \overline{\Gamma} \right)}$ or simply $\mathcal{C}$ if there is no danger of confusion.\\
Then, for each $ g \in \Gamma$, $A \overline{\Gamma}^{T}_{g}$ is the vector subspace of $T_{g} \Gamma$ linearly generated by the evaluation of the admissible vector fields at $g$. Observe that, for all $g \in \Gamma$, the zero vector $0_{g} \in T_{g} \Gamma$ is contained in the fibre of the distribution at $g$, namely $A \overline{\Gamma}^{T}_{g}$ (we remit to \cite{VMMDME,CHARDIST,GENHOM} for non trivial examples). Furthermore, it satisfies that a vector field $\Theta$ of $\Gamma$ holds conditions (i) and (ii) if, and only if, its local flow $\varphi^{\Theta}_{t}$ is left-invariant or, equivalently,
$$ L_{g} \circ \varphi^{\Theta}_{t} = \varphi^{\Theta}_{t} \circ L_{g}, \ \forall g,t.$$
Therefore, condition (iii) is equivalent to the following,
\begin{itemize}
\item[(iii)'] The (local) flow $\varphi^{\Theta}_{t}$ of $\Theta$ at $\overline{g}$ is totally contained in $\overline{\Gamma}$, for all $\overline{g} \in \overline{\Gamma}$.
\end{itemize}
So, the admissible vector fields are the left-invariant vector fields on $\Gamma$ whose integral curves are confined inside or outside $\overline{\Gamma}$.\\
The distribution $A\overline{\Gamma}^{T}$  generated by the vector spaces $A\overline{\Gamma}^{T}_{g}$ is called \textit{characteristic distribution of $\overline{\Gamma}$}. As an immediate result, we have that this distribution is \textit{differentiable}.\\
\begin{remark}
\rm
This construction of the characteristic distribution associated to a subgroupoid $\overline{\Gamma}$ of a Lie groupoid $\Gamma$ may be thought as a generalization of the construction of the associated Lie algebroid to a given Lie groupoid (see \cite{KMG}).\\
\end{remark}

The algebraic structure associated to a groupoid allows us to define more objects. Particularly, one of them is a smooth distribution over the base $M$ denoted by $A \overline{\Gamma}^{\sharp}$. The other is a ``\textit{differentiable}" correspondence $A\overline{\Gamma}$ which associates to any point $x$ of $M$ a vector subspace of $T_{\epsilon \left( x \right) } \Gamma$. Both constructions are characterized by the following diagram

\begin{large}
\begin{center}
 \begin{tikzcd}[column sep=huge,row sep=huge]
\Gamma\arrow[r, "A \overline{\Gamma}^{T}"] &\mathcal{P} \left( T \Gamma \right) \arrow[d, "T\alpha"] \\
 M \arrow[u,"\epsilon"] \arrow[r,"A \overline{\Gamma}^{\sharp}"] \arrow[ru,dashrightarrow, "A \overline{\Gamma}"]&\mathcal{P} \left( T M \right)
 \end{tikzcd}
\end{center}
\end{large}

%

\vspace{15pt}
\noindent{where $\mathcal{P} \left( E \right)$ defines the power set of $E$. Therefore, for any $x \in M$, the fibres are characterized by,}
\begin{eqnarray*}
A \overline{\Gamma}_{x} &=&  A \overline{\Gamma}^{T}_{\epsilon \left( x \right)}\\
A \overline{\Gamma}^{\sharp}_{x}  &=& T_{\epsilon \left( x \right) } \alpha \left( A \overline{\Gamma}_{x} \right)
\end{eqnarray*}
\noindent{The distribution $A \overline{\Gamma}^{\sharp}$ is called \textit{base-characteristic distribution of $\overline{\Gamma}$}.}\\
Notice that, taking into account that $A \overline{\Gamma}^{T}$ is locally generated by left-invariant vector field, we have that for each $g \in \Gamma$,
$$ A \overline{\Gamma}^{T}_{g} = T_{\epsilon \left( \alpha \left( g \right) \right)} L_{g} \left( A \overline{\Gamma}^{T}_{\epsilon \left( \alpha \left( g \right)\right)} \right),$$
i.e., the characteristic distribution is \textit{left-invariant}.}\\
\begin{theorem}[\cite{VMMDME,CHARDIST}]\label{10.24}
Let $\Gamma \rightrightarrows M$ be a Lie groupoid and $\overline{\Gamma}$ be a subgroupoid of $\Gamma$ (not necessarily a Lie groupoid) over $M$. Then, the characteristic distribution $A \overline{\Gamma}^{T}$ is integrable and its associated foliation $\overline{\mathcal{F}}$ of $\Gamma$ satisfies that $\overline{\Gamma}$ is a union of leaves of $\overline{\mathcal{F}}$.
\end{theorem}
This result is a consequence of the celebrated Stefan-Sussman's theorem \cite{PS,HJS} which deals with the integrability of singular distributions.\\
So, the distribution $A \overline{\Gamma}^{T}$ is the tangent distribution of a smooth (possibly) singular foliation $\overline{\mathcal{F}}$. Each leaf at a point $g \in \Gamma$ is denoted by $\overline{\mathcal{F}} \left( g \right)$. Furthermore, the family of the leaves of $\overline{\mathcal{F}}$ at points of $\overline{\Gamma}$ is called the \textit{characteristic foliation of $\overline{\Gamma}$}. Note that the leaves of the characteristic foliation covers $\overline{\Gamma}$ but it is not exactly a foliation of $\overline{\Gamma}$ ($\overline{\Gamma}$ is not necessarily a manifold). The foliation $\overline{\mathcal{F}}$ satisfies that
\begin{itemize}
\item[(i)] For any $g \in \Gamma$,
$$\overline{\mathcal{F}} \left( g \right) \subseteq \Gamma^{\beta \left( g \right)}.$$
Indeed, if $g \in \overline{\Gamma}$, then
$$\overline{\mathcal{F}} \left( g \right) \subseteq \overline{\Gamma}^{\beta \left( g \right)}.$$
\item[(ii)] For any $g ,h \in \Gamma$ such that $\alpha \left( g \right) = \beta \left( h \right)$, we have
$$\overline{\mathcal{F}} \left( g \cdot h\right) = g \cdot \overline{\mathcal{F}} \left(  h\right).$$
\end{itemize}
In this way, without any assumption of differentiability over $\overline{\Gamma}$, we have that $\overline{\Gamma}$ is union of leaves of a foliation of $\Gamma$. This provides some kind of ``\textit{differentiable}" structure over $\overline{\Gamma}$. The following result provides us an intuition about the maximality condition of the characteristic foliation.\\
\begin{corollary}\label{10.33}
Let $\overline{\mathcal{H}}$ be a foliation of $\Gamma$ such that $\overline{\Gamma}$ is a union of leaves of $\overline{\mathcal{H}}$ and
$$\overline{\mathcal{H}} \left(  g \right)  \subset \Gamma^{\beta \left( g \right)}.$$
Then, the characteristic foliation $\overline{\mathcal{F}}$ is coarser that $\overline{\mathcal{H}}$, i.e.,
\begin{equation}\label{10.100}
 \overline{\mathcal{H}} \left( g \right) \subseteq \overline{\mathcal{F}} \left( g \right)  , \ \forall g \in \Gamma.
\end{equation}
\begin{proof}
The result follows from the facts of that $\overline{\mathcal{H}}$ is generated by left-invariant vector fields and any of these left-invariant vector field $\Theta \in T \overline{\mathcal{H}}$ is obviously tangent to the characteristic distribution.
\hfill
\end{proof}
\end{corollary}
\noindent{As a consequence, \textit{the fibres $\overline{\Gamma}^{ x }$ are submanifolds of $\Gamma$ for all $x \in M$ if, and only if, $\overline{\Gamma}^{ x } = \overline{\mathcal{F}} \left( \epsilon \left( x \right) \right)$ for all $x \in M$.}}\\
\begin{proposition}[Consistency]\label{consistencyproperty2904}
Let be $\Gamma \rightrightarrows M$, $\Gamma' \rightrightarrows M'$ two Lie groupoids and $\Phi: \Gamma \rightarrow \Gamma'$ an embedding of Lie groupoids. Consider a (non necessarily Lie) subgroupoid $\overline{\Gamma}$ of $\Gamma$. Then, the image of the characteristic foliation $\overline{\mathcal{F}}$ of $\overline{\Gamma}$ by $\Phi$ is the chacteristic foliation of $\Phi \left( \overline{\Gamma} \right)$ as a subgroupoid of $\Gamma'$.
\begin{proof}
First at all, notice that $\Phi \left(\Gamma \right)$ is Lie groupoid of $\Gamma'$ on $\phi \left( M \right)$, where $\phi$ is projection of $\Phi$ on the base manifolds, because $\Phi$ is an embedding of Lie groupoids.\\
Let $\Theta \in \frak X_{loc} \left( \Gamma \right)$ be an admissible vector field for the couple $\left( \Gamma , \overline{\Gamma} \right)$, i.e.,
\begin{itemize}
\item $\Theta$ is left-invariant.
\item The (local) flow $\varphi^{\Theta}_{t}$ of $\Theta$ satisfies
$$\varphi^{\Theta}_{t} \left( \epsilon \left( x \right)\right) \subseteq \overline{\Gamma}, $$
for all $x \in M$.
\end{itemize}
Then, the pushforward $\Phi_{*}\Theta$ is an admissible vector field for the couple $\left( \Phi \left(\Gamma \right), \Phi \left(\overline{\Gamma} \right)\right)$. In fact, since $\Phi$ is a morphism of Lie groupoids, we have that $\Phi_{*}\Theta$ is left-invariant.\\
On the other hand, the (local) flow of $\Phi_{*}\Theta$ is given by $\Phi \circ \varphi^{\Theta}_{t} \circ \Phi^{-1}$, where $\varphi^{\Theta}_{t}$ is the local flow of $\Theta$. So, at each $x = \phi \left(y\right) \in \phi \left( M \right)$, the local flow of $\Phi_{*}\Theta$ at the identity on $y$, $\Phi \circ \varphi^{\Theta}_{t} \left( \epsilon \left( x \right) \right)$ is totally contained in $\Phi \left( \Gamma \right)$, i.e.,
$$\Phi \circ \varphi^{\Theta}_{t} \left( \epsilon \left( x \right)\right) \in \Phi \left(\overline{\Gamma}\right), \ \forall t.$$
Analogously, given an admissible vector field $\Lambda$ for the couple $\left( \Phi \left(\Gamma \right), \Phi \left(\overline{\Gamma} \right)\right)$, the pushforward $\Phi^{-1}_{*}\Lambda$ is an admissible vector field for the couple $\left( \Gamma , \overline{\Gamma} \right)$. Hence, we have proved that the image of the characteristic foliation $\overline{\mathcal{F}}$ of $\overline{\Gamma}$ by $\Phi$ is the chacteristic foliation of $\Phi \left( \overline{\Gamma} \right)$ as a subgroupoid of $\Phi \left(\Gamma \right)$.\\
Finally, due to the fact that $\Phi \left(\Gamma \right)$ is a Lie subgroupoid of $\Gamma'$ and $\Phi \left( \overline{\Gamma}\right)$ is contained in $\Phi \left(\Gamma \right)$, any admissible vector field $\Lambda$ for the couple $\left( \Phi \left(\Gamma \right), \Phi \left(\overline{\Gamma} \right)\right)$ may be (globally) extended, by using left translations, to an admissible vector field $\tilde{\Lambda}$ for the couple $\left( \Gamma' , \Phi \left(\overline{\Gamma} \right)\right)$. In fact, the extension $\tilde{\Lambda}$ is, by construction, a left-invariant vector field on $\Gamma'$ and its flow at the identities is completely contained in $ \Phi \left(\overline{\Gamma} \right)$. On the other hand, analogously, the restriction to $\Phi \left( \Gamma \right)$ of any admissible vector field $\Theta$ for the couple $\left( \Gamma', \Phi \left(\overline{\Gamma} \right)\right)$ is an admissible vector field for the couple $\left(\Phi\left( \Gamma \right) , \Phi \left(\overline{\Gamma} \right)\right)$. Therefore, the characteristic distribution of $\overline{\Gamma}$ as a subgroupoid of $\Phi \left( \overline{\Gamma}\right)$ is the restriction of the characteristic distribution of $\overline{\Gamma}$ as a subgroupoid of $ \overline{\Gamma}'$.
\end{proof}
\end{proposition}

Thus, this results show a \textit{consistency property} in the definition of the characteristic distribution. In particular, the characteristic foliation (resp. distribution) does not depend on the ``\textit{ambient space}''.\\\\
\noindent{Notice that, analogously to Theorem \ref{10.24}, we may proved that the base-characteristic distribution $A \overline{\Gamma}^{\sharp}$ is integrable. Thus, we will denote the foliation which integrates the base-characteristic distribution over the base $M$ by $\mathcal{F}$. For each point $x \in M$, the leaf of $\mathcal{F}$ containing $x$ will be denoted by $\mathcal{F} \left( x \right)$. $\mathcal{F}$ will be called the \textit{base-characteristic foliation of $\overline{\Gamma}$}.}\\
\begin{example}\label{10.41}
\rm
Let $\sim$ be an equivalence relation on a manifold $M$, i.e., a binary relation that is reflexive, symmetric and transitive. Then, define the subset $\mathcal{O}$ of $M \times M$ given by
\begin{equation}
\mathcal{O}:= \{ \left( x,y \right) \  : \    x \sim y \}.
\end{equation}
Hence, $\mathcal{O}$ is a subgroupoid of $M \times M$ over $M$. In fact, this is equivalent to the properties reflexive, symmetric and transitive. For each $x \in M$, we denote by $\mathcal{O}_{x}$ to the orbit around $x$,
$$\mathcal{O}_{x}:= \{ y  \  : \    x \sim y \}.$$
Notice that the orbits divide $M$ into a disjoint union of subsets. However, these are not (necessarily) submanifolds.\\
On the other hand, the base-characteristic foliation gives us a foliation $\mathcal{F}$ of $M$ such that
$$ \mathcal{F} \left( x \right) \subseteq \mathcal{O}_{x}, \ \forall x \in M.$$
So, consider a random equivalence relation on a manifold $M$. Maybe the orbits are not manifolds but we have proved that \textit{we may divide $M$ in a maximal foliation such that any orbit is a union of leaves.} This foliation is maximal in the sense that there is no any other coarser foliation of $M$ whose leaves are contained in the orbits (see theorem \ref{10.20} and corollary \ref{10.39}).
\end{example}

\noindent{Next, we will show that the leaves of $\mathcal{F}$ may be endowed with even more geometric structure. Indeed, we will construct a Lie groupoid structure over each leaf of $\mathcal{F}$.}\\
For each $x \in M$, let us consider the groupoid $\overline{\Gamma} \left( \mathcal{F} \left( x \right) \right)$ generated by $\overline{\mathcal{F}} \left( \epsilon \left( x \right) \right) $. Notice that, for each $\overline{h} \in \overline{\mathcal{F}} \left( \epsilon \left( x \right) \right)$,
$$  \overline{\mathcal{F}} \left( \epsilon \left( x \right) \right) = \overline{\mathcal{F}} \left(  \overline{h} \right) = \overline{h} \cdot  \overline{\mathcal{F}} \left( \epsilon \left( \alpha \left( \overline{h} \right) \right) \right).$$
Hence,
$$ \overline{\mathcal{F}} \left( \overline{h}^{-1}\right) = \overline{h}^{-1} \cdot  \overline{\mathcal{F}} \left( \epsilon \left( x \right) \right)  =  \overline{\mathcal{F}} \left( \epsilon \left( \alpha \left( \overline{h} \right)\right) \right).$$
On the other hand, let be $\overline{t} \in \overline{\mathcal{F}} \left( \epsilon \left( \alpha \left( \overline{h} \right)\right) \right)$. Therefore,
$$   \overline{\mathcal{F}} \left( \overline{h} \cdot \overline{t} \right) = \overline{h}\cdot  \overline{\mathcal{F}} \left( \overline{t}\right) =  \overline{h}\cdot \overline{\mathcal{F}} \left( \epsilon \left( \alpha \left( \overline{h}\right) \right)\right) =  \overline{\mathcal{F}} \left( \epsilon \left( x \right)\right).$$
i.e., $\overline{h} \cdot \overline{t} \in \overline{\mathcal{F}} \left( \epsilon \left( x \right)\right)$ and, hence, $\overline{t}$ can be written as $\overline{h}^{-1} \cdot \overline{g}$ with $\overline{g} \in \overline{\mathcal{F}} \left( \epsilon \left( x \right)\right).$ So, we have proved that
$$\overline{\mathcal{F}} \left( \epsilon \left( \alpha \left( \overline{h} \right)\right) \right) \subset \overline{\Gamma} \left( \mathcal{F} \left( x \right) \right),$$
for all $\overline{h} \in \overline{\mathcal{F}} \left( \epsilon \left( x\right) \right) $. In fact, by following the same argument we have that
\begin{equation}\label{10.17}
 \overline{\Gamma} \left( \mathcal{F} \left( x \right) \right) = \sqcup_{\overline{g} \in \overline{\mathcal{F}} \left( \epsilon \left( x \right) \right)} \overline{\mathcal{F}} \left( \epsilon \left( \alpha \left( \overline{g} \right) \right) \right),
\end{equation}
i.e., $\overline{\Gamma} \left( \mathcal{F} \left( x \right) \right)$ can be depicted as a disjoint union of fibres at the identities. Furthermore, $\overline{\Gamma} \left( \mathcal{F} \left( x \right) \right)$ may be equivalently defined as the smallest transitive subgroupoid of $\overline{\Gamma}$ which contains $\overline{\mathcal{F}} \left( \epsilon \left( x \right) \right)$. Observe that the $\beta-$fibre of this groupoid at a point $y \in \mathcal{F} \left( x \right)$ is given by $\overline{\mathcal{F}} \left( \epsilon \left( y \right) \right)$. Hence, the $\alpha-$fibre at $y$ is 
$$ \overline{\mathcal{F}}^{-1} \left( \epsilon \left( y\right) \right) = i \circ \overline{\mathcal{F}} \left( \epsilon \left( y \right) \right).$$
Furthermore, the groups $\overline{\mathcal{F}} \left( \epsilon \left( y \right) \right) \cap \Gamma_{ y}$ are exactly the isotropy groups of $\overline{\Gamma}\left( \mathcal{F} \left( x \right) \right) $.
\begin{theorem}\label{10.20}
For each $x \in M$ there exists a transitive Lie subgroupoid $\overline{\Gamma} \left( \mathcal{F} \left( x \right) \right)$ of $\Gamma$ with base $\mathcal{F} \left( x \right)$.

\end{theorem}
The proof of this result comes from some technical lemmas and may be found in \cite{VMMDME,CHARDIST}.\\
Thus, we have divided the manifold $M$ into leaves $\mathcal{F} \left( x \right)$ which have a maximal structure of transitive Lie subgroupoids of $\Gamma$. 

\begin{corollary}[\cite{VMMDME}]\label{10.39}
Let $\mathcal{H}$ be a foliation of $M$ such that for each $x \in M$ there exists a transitive Lie subgroupoid $\Gamma \left( x \right)$ of $\Gamma$ over the leaf $\mathcal{H} \left( x \right)$ contained in $\overline{\Gamma}$ whose family of $\beta-$fibres defines a foliation on $\Gamma$. Then, the base-characteristic foliation $\mathcal{F}$ is coarser than $\mathcal{H}$, i.e.,
$$ \mathcal{H} \left( x \right) \subseteq \mathcal{F} \left( x \right) , \ \forall x \in M.$$
Futhermore it satisfies that
$$ \Gamma \left( x \right) \subseteq \overline{\Gamma} \left( \mathcal{F} \left( x \right) \right).$$
\end{corollary}
As a consequence, we have that $\overline{\Gamma}$ is a transitive Lie subgroupoid of $\Gamma$ if, and only if, $M = \mathcal{F} \left( x \right)$ and $\overline{\Gamma} =\overline{\Gamma} \left( \mathcal{F} \left( x \right) \right)$ for some $x \in M$.\\

Let us consider now the following equivalence relation $\sim$ on $M$ given by 
$$ x \sim y \ \ \ \Leftrightarrow \ \ \ \exists \overline{g} \in \overline{\Gamma}, \ \alpha \left( \overline{g} \right) = x, \ \beta \left( \overline{g} \right) = y.$$
Then, by example \ref{10.41}, we have a subgroupoid $\overline{\Gamma}^{B}$ of the pair groupoid $M \times M$. So, we may consider its associated base-characteristic distribution $A \overline{\Gamma}^{B}$ at $M$ which is called the \textit{transitive distribution of $\overline{\Gamma}$}. The associated base-characteristic foliation $\mathcal{G}$ of $M$ will be called \textit{transitive foliation of $\overline{\Gamma}$}. 
\begin{corollary}\label{10.42}
The base-characteristic foliation $\mathcal{F}$ based on the groupoid $\overline{\Gamma}$ is contained in the transitive foliation $\mathcal{G}$ of $\overline{\Gamma}$.
\begin{proof}
For each $x \in M$, $\mathcal{F} \left( x \right) \times \mathcal{F} \left( x \right)$ defines a transitive Lie subgroupoid of $M \times M$ over $\mathcal{F} \left( x \right)$ and the result follows from corollary \ref{10.39}.
\end{proof}
\end{corollary}
Summarizing, for a fixed subgroupoid $\overline{\Gamma}$ of a Lie groupoid $\Gamma$ we have available three canonical foliations, $\overline{\mathcal{F}}$, $\mathcal{F}$ and $\mathcal{G}$. Roughly speaking, $\mathcal{G}$ divides the base manifold into a maximal foliation such that each leaf is transitive or, in other words, $\mathcal{G}$ divides the orbits of $\overline{\Gamma}$ into a maximal foliation of $M$. The main difference between the foliations $\mathcal{G}$ and $\mathcal{F}$ is that, with $\mathcal{F}$, we are not only requesting ``\textit{differentiability}'' on the base manifold $M$ but on the groupoid $\overline{\Gamma}$.\\
For instance, suppose that $\overline{\Gamma}$ is a transitive subgroupoid of $\Gamma$. Then, $\mathcal{G}$ consists in one unique leaf equal to $M$. However, if $\overline{\Gamma}$ is not a Lie subgroupoid of $\Gamma$ the base-characteristic foliation $\mathcal{F}$ does not have (necessarily) one unique leaf equal to $M$.\\
Apart from Example \ref{10.41}, we may study several relevant applications of the characteristic distribution. In \cite{VMMDME} we may find some of them. Here we are mainly interested in one of them, the so-called \textit{material distributions}, which will be presented in what follows.

\part{Elastic simple materials}\label{partelasticmat}

We will start dealing with the notion of \textit{simple material}. For a detailed introduction to this topic we refer to the books \cite{EPSBOOK,CCWAN,EPSBOOK2}. Another recommendable reference is \cite{JEMARS}.\\
A \textit{(deformable) body} is defined as an oriented manifold $\mathcal{B}$ of dimension $3$ which can be covered by just one chart. The points of the body $\mathcal{B}$ will be called \textit{body points} or \textit{material particles} and will be denoted by using capital letters ($X,Y,Z \in \mathcal{B}$). A \textit{sub-body} of $\mathcal{B}$ is an open subset $\mathcal{U}$ of the manifold $\mathcal{B}$.\\
The existence of the so-called \textit{configurations} arises from the need of manifesting the body into the ``\textit{real world}''. Thus, a configuration is an embedding $\phi : \mathcal{B} \rightarrow \mathbb{R}^{3}$. An \textit{infinitesimal configuration at a particle $X$} is given by the $1-$jet $j_{X,\phi \left(X\right)}^{1} \phi$ where $\phi$ is a configuration of $\mathcal{B}$. The points on the euclidean space $\mathbb{R}^{3}$ will be called \textit{spatial points} and will be denoted by lower case letters ($x,y,z\in \mathbb{R}^{3}$).\\
From now on, we will fix a configuration, denoted by $\phi_{0}$, called \textit{reference configuration}. The image $\mathcal{B}_{0} = \phi_{0} \left( \mathcal{B} \right)$ will be called \textit{reference state}. Coordinates in the reference configuration will be denoted by $X^{I}$, while any other coordinates will be denoted by $x^{i}$.\\
A \textit{deformation} of the body $\mathcal{B}$ is defined as the change of configurations $\kappa = \phi_{1} \circ \phi_{0}^{-1}$ or, equivalently a diffeomorphism from the reference state $\mathcal{B}_{0}$ to any other open subset $\mathcal{B}_{1}$ of $\mathbb{R}^{3}$. Analogously, an \textit{infinitesimal deformation at $\phi_{0}\left(X\right)$} is given by a $1-$jet $j_{\phi_{0}\left(X\right) , \phi \left(X\right)}^{1} \kappa$ where $\kappa$ is a deformation.\\

\noindent{A relevant goal in continuum mechanics is to study the motion of a body. Here, the internal properties of the body will play an important role (gum or rock are not deformed equally under the same loading).}\\
We may interpret this fact as the dymanical principles are not enough to characterize the motion of a deformable body. Thus, following \cite{WNOLLTHE}, the mechanical response of the body to the history of its deformations is supposed to be determined for the so-called \textit{constitutive equations}.\\
For \textit{elastic simple bodies}, or simply \textit{simple bodies}, \cite{CCWAN} we will assume that the constitutive law depends on a particle only on the infinitesimal deformation at the same particle. More explicitly, the \textit{mechanical response} for an (elastic) simple material $\mathcal{B}$, in a fixed reference configuration $\phi_{0}$, is formalized as a differentiable map $W$ from the set $\mathcal{B} \times Gl \left( 3 , \mathbb{R} \right)$, where $Gl \left( 3 , \mathbb{R} \right)$ is the general linear group of $3 \times 3$-regular matrices, to a fixed (finite dimensional) vector space $V$. In general, $V$ will be the space of \textit{stress tensors}.\\
Indeed, in continuum mechanics, the contact forces at a particle $X$ in a given configuration $\phi$ are characterized by a symmetric second-order tensor 
$$T_{X,\phi}: \mathbb{R}^{3} \rightarrow \mathbb{R}^{3}$$
on $\mathbb{R}^{3}$ called the \textit{stress tensor}. A physical interpretation is that $T_{X,\phi}$ turns the unit normal of a smooth surface into the stress vector acting on the surface at $\phi \left( X \right)$. Then, the mechanical response is given by the following identity:
$$W \left( X , F \right) = T_{X,\phi},$$
where $F$ is the $1-$jet at $\phi_{0} \left( X \right)$ of $\phi \circ \phi_{0}^{-1}$.\\
At this point, we should introduce the \textit{rule of change of reference configuration}. In particular, let $\phi_{1}$ be another configuration and $W_{1}$ be the mechanical response associated to $\phi_{1}$. Then, 
\begin{equation}\label{1.5}
 W_{1} \left( X , F \right) = W \left( X , F \cdot C_{01} \right),
\end{equation}
for all regular matrix $F$ where $C_{01}$ is the associated matrix to the $1-$jet at $\phi_{0} \left( X \right)$ of $\phi_{1} \circ \phi_{0}^{-1}$. Equivalently,
\begin{equation}\label{1.4}
 W \left( X , F_{0} \right) = W_{1} \left( X , F_{1} \right),
\end{equation}
where $F_{i}$, $i=0,1$, is the associated matrix to the $1-$jet at $\phi_{i} \left( X \right)$ of $\phi \circ \phi_{i}^{-1}$ with $\phi$ a configuration.
It is important to remark that Eq. (\ref{1.5}) implies that we may define $W$ as a map on the space of $1$-jets of (local) configurations which is \textit{independent on the chosen reference configuration}. In fact, for each configuration $\phi$ we will define
$$ W \left( j^{1}_{X, x} \phi \right) = W \left( X , F \right),$$
where $F$ is the associated matrix to the $1-$jet at $\phi_{0} \left( X \right)$ of $\phi \circ \phi_{0}^{-1}$.\\
Notice that any sub-body inherits the  structure of elastic simple body from the body $\mathcal{B}$. This local property permits us to compare the material properties at the particles of the body. In particular, we may study when two particles $X$ and $Y$ are made of the same material. To do this, we will introduce the notion of \textit{material isomorphisms}.

\begin{definition}\label{1.33}
\rm
Let $\mathcal{B}$ be a body. Two material particles $X , Y \in \mathcal{B}$ are said to be \textit{materially isomorphic} if there exists a local diffeomorphism $\psi$ from an open neighbourhood $\mathcal{U} \subseteq \mathcal{B}$ of $X$ to an open neighbourhood $\mathcal{V} \subseteq \mathcal{B}$ of $Y$ such that $\psi \left(X\right) =Y$ and
\begin{equation}\label{1.3.first}
W \left( X , F \cdot P \right) = W \left( Y , F \right),
\end{equation}
for all infinitesimal deformation $F$ where $P$ is given by the Jacobian matrix of $\phi_{0} \circ \psi \circ \phi_{0}^{-1}$ at $\phi_{0} \left( X \right)$. The $1-$jets of local diffeomorphisms satisfying Eq. (\ref{1.3.first}) are called \textit{material isomorphisms}. A material isomorphism from $X$ to itself is called a \textit{material symmetry}. In cases where it causes no confusion we often refer to associated matrix $P$ as the material isomorphism (or symmetry).
\end{definition}
So, intuitively, two points are materially isomorphic if the constitutive equation of one of them differs from the other only by an application of a linear transportation, i.e., the are made of the same material.\\
It is remarkable that the relation of being ``\textit{materially isomorphic}" defines an equivalence relation (symmetric, reflexive and transitive) over the body manifold $\mathcal{B}$.\\
We will denote by $G \left( X \right)$ to the set of all material symmetries at particle $X$. As a consequence we have that every $G \left( X \right)$ is a group. Therefore, we may prove that the material symmetry groups of materially isomorphic particles are conjugated, i.e., if $X$ and $Y$ are material isomorphic we have that
$$ G \left( Y \right) = P \cdot G \left( X \right) \cdot P^{-1},$$
where $P$ is a material isomorphism from $X$ to $Y$.
\begin{proposition}
Let $\mathcal{B}$ be a body. Two body points $X$ and $Y$ are materially isomorphic if, and only if, there exist two (local) configurations $\phi_{1}$ and $\phi_{2}$ such that
$$ W_{1} \left( X , F \right)= W_{2} \left( Y , F \right), \ \forall F,$$
where $W_{i}$ is the mechanical response associated to $\phi_{i}$ for $i=1,2$.
\begin{proof}
Let $j_{X,Y}^{1}\psi$ be a material isomorphism from $X$ to $Y$. Then, we may prove the result by imposing,
$$\phi_{2} = \phi_{1}\circ \psi,$$
where $\phi_{i}$ is the reference configuration for $W_{i}$.

\end{proof}
\end{proposition}
This result is crucial to understand the idea of material isomorphism. Thus, two material points will be made of the same material if their mechanical responses are the same up to a change of reference configuration.\\

\begin{definition}\label{1.17}
\rm
A body $\mathcal{B}$ is said to be \textit{uniform} if all of its body points are materially isomorphic.
\end{definition}
Intutively, a body is uniform if all the points are made of the same material. Let $\mathcal{B}$ be a uniform body and a fixed body point $X_{0}$; for any other body point $Y$ we may find a material isomorphism from $X_{0}$ to $Y$, say $P \left( Y \right) \in Gl \left( 3, \mathbb{R} \right)$. Then, we shall construct a map $P : \mathcal{B} \rightarrow   Gl \left( 3, \mathbb{R} \right)$ consisting of material isomorphisms. However, $P$ is not in general a differentiable map.
\begin{definition}\label{1.7}
\rm
A body $\mathcal{B}$ is said to be \textit{smoothly uniform} if for each point $X \in \mathcal{B}$ there is a neighbourhood $\mathcal{U} $ around $X$ and a smooth map $P : \mathcal{U} \rightarrow Gl \left( 3, \mathbb{R} \right)$ such that for all $Y \in \mathcal{U}$ it satisfies that $P \left(Y\right)$ is a material isomorphism from $X$ to $Y$. The map $P$ is called a \textit{right (local) smooth field of material isomorphisms}. A \textit{left (local) smooth field of material isomorphisms} is defined analogously.
\end{definition}
Let $P$ be a right (local) smooth field of material isomorphisms. Hence, the mechanical response of the sub-body $\mathcal{U}$ satisfies that
$$W \left( Y , F \right) = W \left( X , F \cdot P \left( Y \right)\right),$$
for all $Y \in \mathcal{U}$. Then, we may define
$$ \overline{W} \left( F \right) = W \left( X , F \right).$$
Therefore,
\begin{equation}\label{1.8}
W \left( Y , F \right) = \overline{W}\left( F \cdot P \left( Y \right)\right).
\end{equation}
Eq. (\ref{1.8}) is interpreted as that the dependence of the mechanical response (near to a material particle) of the body is given by a multiplication of $F$ to the right by a right smooth field of material isomorphisms.
\begin{proposition}\label{1.18}
Let $\mathcal{B}$ be a body. Then, $\mathcal{B}$ is (smoothly) uniform if, and only if, there exists a (differentiable) map $\overline{W}: Gl \left( 3, \mathbb{R} \right) \rightarrow V$ satisfying Eq. (\ref{1.8}) for a (differentiable) map $P : \mathcal{U} \rightarrow Gl \left( 3, \mathbb{R} \right)$.
\begin{proof}
Assume that Eq. (\ref{1.8}) is satisfied for a map $P$ and fix a material point $X$. Then, consider
$$Q : \mathcal{U} \rightarrow Gl \left( 3, \mathbb{R} \right)$$
given by
\begin{equation}\label{AnotherEqMI34}
Q \left( Y  \right) = P \left(Y \right) P\left(  X  \right)^{-1}.
\end{equation}
Therefore, $Q$ is a left (smooth) field of material isomorphisms.
\end{proof}
\end{proposition}
It is important to note that the smooth uniformity is the starting point of the use of $G-$structures in \cite{MELZA} (see \cite{MELZASEG} or \cite{CCWANSEG}; see also \cite{FBLOO} and \cite{GAMAU2}). In fact, let us consider a smoothly uniform body $\mathcal{B}$. Fix $Z_{0} \in \mathcal{B}$ and $\overline{Z}_{0} = j^{1}_{0,Z_{0}} \phi \in F \mathcal{B}$ a frame at $Z_{0}$. Then, the following set:

$$\omega_{G_{0}} \left( \mathcal{B}\right) := \{ j^{1}_{Z_{0},Y}\psi \cdot \overline{Z}_{0}, \ : \  j^{1}_{Z_{0},Y}\psi \ \text{is a material isomorphism}   \},$$ is a $G_{0}-$structure on $\mathcal{B}$ (which contains $\overline{Z}_{0}$). This $G_{0}-$structure has been used to study simple material. However, it is defined only for smoothly uniform materials and it is not canonically defined.\\

\noindent{The use of groupoids solved these two points as may be found in \cite{VMJIMM,MATGROUPALG} (see also \cite{MGEOEPS,COSVME}). Let $\mathcal{B}$ be a elastic simple body with reference configuration $\phi_{0}$, and mechanical response $W: \mathcal{B} \times Gl \left( 3, \mathbb{R} \right) \rightarrow V$. Eq. (\ref{1.4}) permits us to define $W$ on the space of (local) configurations in such a way that for any configuration $\phi$ we have that}
$$ W \left( j^{1}_{X, x} \phi \right) = W \left( X , F \right),$$
where $F$ is the associated matrix to the $1-$jet at $\phi_{0} \left( X \right)$ of $\phi \circ \phi_{0}^{-1}$. Indeed, composing $\phi_{0}$ by the left, we obtain that $W$ may be described as a differentiable map $W : \Pi^{1}\left( \mathcal{B} , \mathcal{B} \right) \rightarrow V$ from the groupoid of $1-$jets $\Pi^{1}\left( \mathcal{B} , \mathcal{B} \right)$ (see example \ref{8}) to the vector space $ V$ which does not depend on the image point of the $1-$jets of $\Pi^{1}\left( \mathcal{B} , \mathcal{B} \right)$, i.e., for all $X,Y,Z \in \mathcal{B}$
\begin{small}

\begin{equation}\label{4.1}
 W \left( j_{X,Y}^{1} \phi\right) = W \left( j_{X,Z}^{1} \left( \phi_{0}^{-1}\circ \tau_{Z-Y} \circ \phi_{0} \circ \phi\right)\right),
\end{equation}

\end{small}
\noindent{for all $j_{X,Y}^{1}\phi \in \Pi^{1} \left( \mathcal{B}, \mathcal{B}\right)$, where $\tau_{v}$ is the translation map on $\mathbb{R}^{3}$ by the vector $v$.} It is relevant to note here that, in contrast with the definition on the space of local configuration, the definition of the mechanical response on $\Pi^{1}\left( \mathcal{B} , \mathcal{B} \right)$ does depend on the choice of the configuration $\phi_{0}$.\\
Therefore, we can say that, \textit{two material particles $X$ and $Y$ are materially isomorphic if, and only if, there exists a local diffeomorphism $\psi$ from an open subset $\mathcal{U} \subseteq \mathcal{B}$ of $X$ to an open subset $\mathcal{V} \subseteq \mathcal{B}$ of $Y$ such that $\psi \left(X\right) =Y$ and
\begin{equation}\label{4.2}
W \left( j^{1}_{Y, \kappa \left(Y\right)} \kappa \cdot j^{1}_{X,Y} \psi \right) = W \left( j^{1}_{Y, \kappa \left(Y\right)} \kappa\right),
\end{equation}
for all $j^{1}_{Y , \kappa \left(Y\right)} \kappa \in \Pi^{1}\left( \mathcal{B} , \mathcal{B} \right)$.} In these conditions, $j^{1}_{X,Y} \psi$ will be called a material \textit{isomorphism from $X$ to $Y$.}\\
For any two points $X,Y  \in \mathcal{B}$, the collection of all material isomorphisms from $X$ to $Y$ will be denoted by $G \left(X,Y\right)$. Then, the set
\begin{equation}\label{materialgroupoid232}
\Omega \left( \mathcal{B}\right) = \cup_{X,Y \in \mathcal{B}} G\left(X,Y\right). 
\end{equation}
is a subgroupoid of $\Pi^{1}\left( \mathcal{B} , \mathcal{B} \right)$. This groupoid will be called \textit{material groupoid of} $\mathcal{B}$.\\
The material symmetry group $G \left( X\right)$ at a body point $X \in \mathcal{B}$ is simply the isotropy group of $\Omega \left( \mathcal{B}\right)$ at $X$. For any $X\in \mathcal{B}$, the set of material isomorphisms from $X$ to any other point (resp. from any point to $X$) will be denoted by $\Omega\left( \mathcal{B}\right)_{X}$ (resp. $\Omega\left( \mathcal{B}\right)^{X}$). Finally, the structure maps of $\Omega \left( \mathcal{B} \right)$ will be denoted by $\overline{\alpha}$, $\overline{\beta}$, $\overline{\epsilon}$ and $\overline{i}$ which are just the restrictions of the corresponding ones on $\Pi^{1} \left(  \mathcal{B} , \mathcal{B} \right)$.\\
As a consequence of the continuity of $W$ we have that, for all $X \in \mathcal{B}$, $G \left( X \right)$ is a closed subgroup of $\Pi^{1} \left(\mathcal{B} , \mathcal{B}\right)_{X}^{X}$ and, therefore, we have the following result:
\begin{proposition}\label{4.3}
Let $\mathcal{B}$ be a simple body. Then, for all $X \in \mathcal{B}$ the symmetry group $G \left( X \right)$ is a Lie subgroup of $\Pi^{1} \left( \mathcal{B} , \mathcal{B}\right)_{X}^{X}$.
\end{proposition}
This result could give us the intuition of that $\Omega \left( \mathcal{B} \right)$ is a Lie subgroupoid of $\Pi^{1} \left( \mathcal{B} , \mathcal{B}\right)$. However, this is not true (see  \cite{CHARDIST,GENHOM} for some counterexamples).
\begin{proposition}
Let $\mathcal{B}$ be a body. $\mathcal{B}$ is uniform if and only if $\Omega \left( \mathcal{B}\right)$ is a transitive subgroupoid of $\Pi^{1} \left( \mathcal{B} , \mathcal{B}\right)$.
\end{proposition}
Next, by composing appropriately with the reference configuration, smooth uniformity (Definition \ref{1.7}) may be characterized in the following way.
\begin{proposition}\label{4.40}
A body $\mathcal{B}$ is smoothly uniform if, and only if, for each point $X \in \mathcal{B}$ there is an neighbourhood $\mathcal{U} $ around $X$ such that for all $Y \in \mathcal{U}$ and $j_{Y,X}^{1} \phi \in \Omega \left( \mathcal{B} \right)$ there exists a local section $\mathcal{P}$ of
$$ \overline{\alpha}_{X} :\Omega\left( \mathcal{B}\right)^{X} \rightarrow \mathcal{B},$$
from $\epsilon \left( X \right)$ to $j_{Y,X}^{1} \phi$.
\end{proposition}
For obvious reasons, (local) sections of $\overline{\alpha}_{X}$ will be called \textit{left fields of material isomorphism at $X$}. On the other hand, local sections of
$$ \overline{\beta}^{X} :\Omega\left( \mathcal{B}\right)_{X} \rightarrow \mathcal{B},$$
will be called \textit{right fields of material isomorphism at $X$}.\\
So, $\mathcal{B}$ is smoothly uniform if, and only if, for any two particles $X,Y \in \mathcal{B}$ there are two open neighbourhoods $\mathcal{U} , \mathcal{V} \subseteq \mathcal{B}$ around $X$ and $Y$ respectively and $\mathcal{P} : \mathcal{U} \times \mathcal{V}  \rightarrow \Omega \left( \mathcal{B} \right) \subseteq \Pi^{1} \left( \mathcal{B} , \mathcal{B} \right)$, a differentiable section of the anchor map $\left( \overline{\alpha} , \overline{\beta} \right)$. When $X=Y$ we may assume $\mathcal{U} = \mathcal{V}$ and $\mathcal{P}$ is a morphism of groupoids over the identity map, i.e.,

$$ \mathcal{P} \left( Z, T \right) = \mathcal{P} \left(  R , T\right) \mathcal{P} \left( Z, R \right), \ \forall T,R,Z \in \mathcal{U}.$$ 

\noindent{So, we have the following corollary of proposition \ref{4.3}.}

\begin{corollary}\label{4.4}
Let $\mathcal{B}$ be a body. $\mathcal{B}$ is smoothly uniform if and only if $\Omega \left( \mathcal{B}\right)$ is a transitive Lie subgroupoid of $\Pi^{1} \left( \mathcal{B} , \mathcal{B}\right)$.
\begin{proof}
Assume that $\mathcal{B}$ is smoothly uniform. Consider $j_{X,Y}^{1} \psi \in \Omega \left( \mathcal{B} \right)$ and $\mathcal{P} : \mathcal{U} \times \mathcal{V} \rightarrow \Omega \left( \mathcal{B} \right)$, a differentiable section of the anchor map $\left( \overline{\alpha} , \overline{\beta} \right)$ with $X \in \mathcal{U}$ and $Y \in \mathcal{V}$. Then, we may construct the following bijection

$$
\begin{array}{rccl}
\Psi_{\mathcal{U} , \mathcal{V}} : & \Omega \left( \mathcal{U} , \mathcal{V}\right) & \rightarrow & \mathcal{B} \times \mathcal{B} \times G \left( X , Y \right)\\
&j_{Z,T}^{1} \phi &\mapsto &  \left( Z,T , \mathcal{P} \left( Z,Y \right) \left[ j_{Z,T}^{1} \phi \right]^{-1} \mathcal{P} \left( X,T \right)\right)
\end{array}
$$

\noindent{where $\Omega \left( \mathcal{U} , \mathcal{V}\right)$ is the set of material isomorphisms from $\mathcal{U}$ to $\mathcal{V}$. By using proposition \ref{4.3}, we deduce that $G \left( X , Y \right)$ is a differentiable manifold. Thus, we can endow $\Omega \left( \mathcal{B} \right)$ with a differentiable structure of a manifold. Finally, the converse has been proved in \cite{KMG}).}
\end{proof}
\end{corollary}
This corollary is useful to understand the difference between smooth uniformity and ordinary uniformity. Furthermore, it provides an intuition about the lack of differentiability which could have the material groupoid. Thus, it arises the need of using the characteristic distribution (see the previous section).\\
Consider $\mathcal{B}$ as a simple body with $W : \Pi^{1}\left( \mathcal{B} , \mathcal{B} \right) \rightarrow V$ as the mechanical response. Then, we have available the so-called material groupoid $\Omega \left( \mathcal{B} \right)$ which is a (non necessarily Lie) subgroupoid of the groupoid of $1-$jets $\Pi^{1}\left( \mathcal{B} , \mathcal{B} \right)$. So, it makes sense to apply here the notion of characteristic distribution.\\
Let $\Theta$ be an admissible vector field for the couple $\left(\Pi^{1} \left( \mathcal{B} , \mathcal{B} \right), \Omega \left( \mathcal{B} \right)\right)$, i.e., its local flow at the identity $\epsilon \left( X \right)$, $\varphi^{\Theta}_{t}\left( \epsilon \left( X \right) \right)$, satisfies that
$$\varphi^{\Theta}_{t}\left( \epsilon \left( X \right) \right) \subseteq \Omega \left( \mathcal{B} \right)$$
for all $X \in \mathcal{B}$ and $t$ in the domain of the flow at $\epsilon \left(X \right)$. Therefore, for any $g \in \Pi^{1} \left( \mathcal{B} , \mathcal{B} \right)$, we have
\begin{eqnarray*}
TW \left( \Theta \left( g \right)\right) &=&  \dfrac{\partial}{\partial t_{\vert 0}}\left(W \left(\varphi^{\Theta}_{t}\left( g \right) \right)\right) \\
&=&  \dfrac{\partial}{\partial t_{\vert 0}}\left(W \left(g \cdot\varphi^{\Theta}_{t}\left( \epsilon \left( \alpha \left( g \right) \right) \right) \right)\right)  \\
&=&  \dfrac{\partial}{\partial t_{\vert 0}}\left(W \left(g \right)\right)  = 0.
\end{eqnarray*}
Hence, we obtain
\begin{equation}\label{14.22}
TW \left( \Theta \right) = 0
\end{equation}
The converse is proved in a similar way.\\
So, the characteristic distribution $A \Omega \left( \mathcal{B} \right)^{T}$ of the material groupoid will be called \textit{material distribution} and it is generated by the (left-invariant) vector fields on $\Pi^{1} \left( \mathcal{B} , \mathcal{B} \right)$ which are in the kernel of $TW$. The base-characteristic distribution $A \Omega \left( \mathcal{B} \right)^{\sharp}$ (see Theorem \ref{10.24}) will be called \textit{body-material distribution} and the transitive distribution will be called \textit{uniform-material distribution}.\\
The foliations associated to the material distribution, the body-material distribution and uniform-material distribution will be called \textit{material foliation}, \textit{body-material foliation} and \textit{uniform-material foliation} and they will be denoted by $\overline{\mathcal{F}} $, $\mathcal{F}$ and $\mathcal{G}$, respectively.\\
For each $X \in \mathcal{B}$, we will denote the Lie groupoid $\Omega \left( \mathcal{B} \right)\left(\mathcal{F}\left( X \right)\right)$ by $\Omega \left( \mathcal{F} \left( X \right) \right)$ (see theorem \ref{10.20}). Denote the groupoid of all material isomorphisms at points in $\mathcal{G}\left( X \right)$ by $\Omega \left( \mathcal{G} \left( X \right) \right)$. Recall that $\Omega \left( \mathcal{F} \left( X \right) \right)$ is a subgroupoid of  $\Omega \left( \mathcal{G} \left( X \right) \right)$, i.e., $\Omega \left( \mathcal{F} \left( X \right) \right) \leq \Omega \left( \mathcal{G} \left( X \right) \right)$. In fact, in the general case, \textit{the condition of maximality on the leaves of $\mathcal{G}$ means that $\mathcal{G}$ is the coarsest foliation such that, at each leaf $\mathcal{G} \left( X \right)$, the groupoid of all material isomorphisms at points in $\mathcal{G}\left( X \right)$ is a transitive subgroupoid of $\overline{\Gamma}$.}\\
Observe that, in continuum mechanics a \textit{sub-body} of a body $\mathcal{B}$ is given by an open submanifold of $\mathcal{B}$ but, here, the foliation $\mathcal{F}$ gives us submanifolds of different dimensions (not only dimension 3). Thus, we will follow \cite{CHARDIST,GENHOM} for a more general definition
\begin{definition}\label{materialsubm45623}
\rm
A \textit{material submanifold (or \textit{generalized sub-body}) of $\mathcal{B}$} is a submanifold of $\mathcal{B}$.
\end{definition}
It is important to note that any generalized sub-body $\mathcal{P}$ inherits certain material structure from $\mathcal{B}$. Particularly, the material response of a material submanifold $\mathcal{P}$ is measured by restricting $W$ to the $1-$jets of local diffeomorphisms $\phi$ on $\mathcal{B}$ from $\mathcal{P}$ to $\mathcal{P}$. However, it is easy to observe that a material submanifold of a body is not exactly a body. See \cite{MD} for a discussion on this subject.\\
Then, as a corollary of Theorem \ref{10.24} and corollary \ref{10.39}, we have the following result.
\begin{theorem}\label{14.1}
The body-material foliation $\mathcal{F}$ (resp. uniform material foliation $\mathcal{G}$) divides the body $\mathcal{B}$ into maximal smoothly uniform material submanifolds (resp. uniform material submanifolds).
\end{theorem}
It should be observed that, in this case, ``maximal'' means that any other foliation $\mathcal{H}$ by smoothly uniform material submanifolds (resp. uniform material submanifolds) is thinner than $\mathcal{F}$ (resp. $\mathcal{G}$), i.e.,
$$ \mathcal{H} \left( X \right) \subseteq  \mathcal{F} \left( X \right) \left(\text{resp. } \mathcal{G}\left( X \right) \right)  , \ \forall X \in  \mathcal{B}.$$
Therefore, the application of material distributions has been used to prove this very intuitive result: \textit{Let $\mathcal{B}
$ a general (smoothly uniform or not) simple material. Then, $\mathcal{B}$ may be decomposed into ``\textit{(smoothly) uniform parts}'' and this decomposition is, in fact, a foliation of the material body.}\\
The material distributions are useful to define new notions la \textit{graded uniformity} and \textit{generalized homogeneity} (see \cite{GENHOM}). However, here we are interested in another way of apply the characteristic distributions. In particular, we want to study the notion of \textit{material evolution}.

\part{Material evolution}
\section{Body-time manifolds and material isomorphisms}
Now, we will present the evolution of the body along time mainly following the references \cite{EPSTEIN201572,MEPMDLSEG,EPSBOOK2}. In our geometrical description of the theory of simple bodies, the time has not played a role. Our body is, in some sense, frozen. Nevertheless, it happens that in some practical
applications, the material properties of the body may change with time. A relevant example is given by the volumetric growth and remodeling of biological tissues, such as bone and muscle.\\
Thus, material evolution is, roughly speaking, the temporal counterpart of the notion of material body. In the case of a material body, we compare the constitutive response of two different material particles at the same instant of time. On the other hand, in material evolution we study the constitutive properties of different points at different instants of time.\\
Then, we consider a \textit{body-time manifold} as the fibre bundle $ \mathcal{C} = \mathbb{R} \times \mathcal{B}$ over $\mathbb{R}$. By simplicity, time and space are supposed to be absolute, but may be easily generalized to a general case.
\begin{definition}\label{history232}
\rm
A \textit{history} is given by a fibre bundle embedding $\Phi : \mathcal{C} \rightarrow \mathbb{R} \times \mathbb{R}^{3}$ over the identity.
\end{definition}
Equivalently, $\Phi$ can be seen as a differentiable family of configurations $\phi_{t} : \mathcal{B} \rightarrow \mathbb{R}^{3}$ such that
\begin{equation}\label{historyequation233}
  \phi\left( t,  X \right) =\phi_{t} \left( X \right) = \left( pr_{2}\circ \Phi\right) \left( t , X \right), \ \forall t \in \mathbb{R}, \ \forall X \in \mathcal{B},  
\end{equation}
where $ pr_{2}\ : \ \mathbb{R} \times \mathbb{R}^{3} \rightarrow \mathbb{R}^{3}$ is the projection on the second component.\\ 
In this way $\Phi$ represent the evolution of the body in time $t$ in such a way that the configuration of $\mathcal{B}$ at time $t$ is $\phi_{t}$. Then, at each instant of time $t$, one may consider the infinitesimal configuration at time $t$, $1-$jet $j_{X, \phi_{t}\left(X \right)}^{1} \phi_{t}$.\\
Next, we need to introduce the constitutive law of the material evolution. In the framework of simple bodies, we will assume that, for a fixed reference configuration $\phi_{0}$, the constitutive response at each material particle $X$ and at each instant of time $t$ may be characterized by one (or more) functions depending on the associated matrices $F$ to the infinitesimal configurations $j_{X, \phi_{t}\left(X \right)}^{1} \phi_{t}$ at particle $X$ and time $t$. So, the \textit{mechanical response} will be a differentiable map,
$$ W :  \mathcal{C} \times Gl\left(3 , \mathbb{R}\right)  \rightarrow V,$$
where $V$ is again a real vector space (generally, $V$ will be assumed to be the space of stress tensors). The definition of the mechanical response permit us to compare material responses at different particles at different instants of time.\\
Once again, the construction of the mechanical response seems to be constrained to the fixed reference configuration. To clarify this dependence we have the \textit{rule of change of reference configuration}.\\
Thus, consider a different configuration $\phi_{1}$ and $W_{1}$ its associated mechanical response. Then, it will be imposed that
\begin{equation}\label{1.5.2}
 W_{1} \left(t, X , F \right) = W \left( t, X , F \cdot C_{01} \right),
\end{equation}
for all regular matrix $F$ where $C_{01}$ is the associated matrix to the $1-$jet at $\phi_{0} \left( X \right)$ of $\phi_{1} \circ \phi_{0}^{-1}$. Equivalently,
\begin{equation}\label{1.4.2}
 W \left( t,X , F_{0} \right) = W_{1} \left(t, X , F_{1} \right),
\end{equation}
where $F_{i}$, $i=0,1$, is the associated matrix to the $1-$jet at $\phi_{i} \left( X \right)$ of $\phi \circ \phi_{i}^{-1}$ with $\phi$ a configuration.\\
Therefore Eq. (\ref{1.5.2}) permit us to define $W$ over the space of (local) histories which is \textit{independent on the chosen reference configuration}. In fact, for each history $\Phi = \phi_{t}$ we will define
\begin{equation}\label{1.4.5.5}
    W \left( t , X , \Phi \right) = W \left( j^{1}_{X,x} \phi_{t} \right) = W \left( t, X , F_{t} \right),
\end{equation}
where $F_{t}$ is the associated matrix to the $1-$jet $j^{1}_{\phi_{0} \left( X \right),x}\left( \phi_{t} \circ \phi_{0}^{-1}\right)$ at $\phi_{0} \left( X \right)$. Reciprocally, for each point $\left( t , X \right)$ any differentiable map $W_{\left( t , X \right)}$ on the space of (local) histories defines a constitutive functional at $\left( t , X\right)$ by Eq. (\ref{1.4.5.5}) satisfying the rule of change of reference configuration (\ref{1.5.2}).\\
Observe that, for all $t$ the manifold $\{t\}\times \mathcal{B}$ inherits the structure of simple body by restricting the mechanical response $W$ to the history of deformations at the same instant $t$. This body will be called \textit{state $t$ of the body $\mathcal{B}$} and it will be denoted by $\mathcal{B}_{t}$. As long as it invites no confusion, we will refer to the simple body $\{0\} \times \mathcal{B}$ as the material body $\mathcal{B}$.\\
On the other hand, it is also important to say that the mechanical response defines an structure of material evolution on any sub-body $\mathcal{U}$ of the body $\mathcal{B}$ by restriction. Nevertheless, analogously to Definition \ref{materialsubm45623}, we will need to relax the definition of ``\textit{material evolution}'' to permit variation of material submanifolds along time.

\begin{definition}\label{materialsubevol56}
\rm
An \textit{evolution material} for a \textit{material submanifold} (or \textit{body-time generalized sub-body}) of $\mathcal{C}$ is a submanifold $\mathcal{M}$ of $\mathcal{C}$.
\end{definition}
Thus, for each instant $t$ (such that there exists a particle $X$ with $\left( t , X \right) \in \mathcal{M}$) we have that the \textit{state $t$ of the material submanifold} is 
$$\left(\{t\}\times \mathcal{B}  \right) \cap \mathcal{M} = \{t\}\times \mathcal{M}_{t},$$
for a submanifold $\mathcal{M}_{t}$ of $\mathcal{B}$.
Hence, varying $t$ we may see how the material submanifold $\mathcal{M}_{t}$ changes along the time. Notice that, we do no impose that 
$$\mathcal{M} = I \times \mathcal{N},$$
for an interval $I$ and a material submanifold $\mathcal{N}$ because we are permitting variations in the ``\textit{shape}'' of $\mathcal{N}$.
\begin{definition}\label{1.33.2}
\rm
Let $\mathcal{C}$ be a body-time manifold. Two pairs $\left( t, X \right), \left( s ,Y\right) \in \mathcal{C}$ are said to be \textit{materially isomorphic} if there exists a local diffeomorphism $\psi$ from an open neighbourhood $\mathcal{U} \subseteq \mathcal{B}$ of $X$ to an open neighbourhood $\mathcal{V} \subseteq \mathcal{B}$ of $Y$ such that $\psi \left(X\right) =Y$ and
\begin{equation}\label{1.3}
W \left( t, X , F \cdot P \right) = W \left( s, Y , F \right),
\end{equation}
for all infinitesimal deformation $F$ where $P$ is given by the Jacobian matrix of $\phi_{0} \circ \psi \circ \phi_{0}^{-1}$ at $\phi_{0} \left( X \right)$. The $1-$jets of local diffeomorphisms satisfying Eq. (\ref{1.3}) are called \textit{time-material isomorphisms} (or \textit{material isomorphisms} if there is no danger of confusion) from $\left(t,X\right)$ to $\left(s,Y\right)$ . A material isomorphism from $\left(t,X\right)$ to itself is called a \textit{time-material symmetry} or \textit{material symmetry}.
\end{definition}
Roughly speaking, two pairs $\left( t, X \right), \left( s ,Y\right) \in \mathcal{C}$ are materially isomorphic if the material points $X$ and $Y$ are made of the same material at the instants $t$ and $s$ respectively. As a particular case, we may consider that $\left( t, X \right)$ and $ \left( s ,X\right)$ are materially isomorphic for all $t$ and $s$ in an interval $I$. Then, the constitutive properties of the material particle $X$ do not change in the time interval $I$.\\
We will denote by $G \left(t, X \right)$ to the set of all material symmetries at $\left(t,X\right)$. Again, any $G \left( t,X \right)$ is a group.
\begin{proposition}
Let $\mathcal{C}$ be a body-time manifold. Two body pairs $\left(t,X\right)$ and $\left(s,Y\right)$ are materially isomorphic if, and only if, there exist two (local) configurations $\phi_{1}$ and $\phi_{2}$ such that
$$ W_{1} \left( t,X , F \right)= W_{2} \left( s,Y , F \right), \ \forall F,$$
where $W_{i}$ is the mechanical response associated to $\phi_{i}$ for $i=1,2$.
\begin{proof}

Consider $j_{X,Y}^{1}\psi$ a material isomorphism from $\left(t,X\right)$ to $\left(s,Y\right)$. We will choose $\phi_{1} = \phi_{0}$ and
$$\phi_{2} = \phi_{1}\circ \psi$$
Then,
\begin{eqnarray*}
W_{1} \left( t,X , F \right) &=& W_{1} \left( s,Y, F \cdot P \right)\\
 &=& W_{2} \left( s,Y, F \cdot P \cdot C_{21}\right)\\
 &=& W_{2} \left( s,Y, F \right)
\end{eqnarray*}
where $P$ is the Jacobian matrix of $\phi_{1} \circ \psi \circ \phi_{1}^{-1}$ at $\phi_{1} \left( X \right)$.
\end{proof}
\end{proposition}
So, analogously to the case in which the body does not depend on time, this result proves the intuitive idea of that two points are materially isomorphic if their constitutive properties are the equal.

\section{Evolution material groupoids}

Now, let us consider the vertical subbundle associated to the body-time manifold $\mathcal{C}$, $\mathcal{V}$, and the associated frame groupoid (see Example \ref{8}) $\Phi \left(\mathcal{V}\right) \rightrightarrows \mathcal{C}$. Notice that, for all $\left( t, X \right) \in \mathcal{C}$, we have that
$$ \mathcal{V}_{\left( t, X \right)} = \{0\} \times T_{X}\mathcal{B}.$$
At this point, it is important to highlight that the groupoid $\Phi \left(\mathcal{V}\right) \rightrightarrows \mathcal{C}$ will be relevant in what follows. In fact, \textit{the role of $\Phi \left(\mathcal{V}\right) \rightrightarrows \mathcal{C}$ for material evolution is comparable to the role of the $1-$jets groupoid $\Pi^{1}\left( \mathcal{B} , \mathcal{B} \right)$ on $\mathcal{B}$ for elastic simple material} (see Part \ref{partelasticmat}).\\\\
The reader could now considering a natural question: \textit{why do we take this groupoid instead of $\Pi^{1}\left( \mathcal{C}, \mathcal{C}\right)$ or any other subgroupoid of this one?}\\
The answer is simple: the elements of this groupoids may be identifyed with the $1-$jets of the so-called histories (see definition \ref{history232}) via a reference configuration.\\\\
Let $\Phi: \mathcal{C} \rightarrow \mathcal{C}$ a (local) embedding of fibre bundles over the identity on $\mathbb{R}$. Then, an element of $\Phi \left(\mathcal{V}\right)$ may be given by a triple $\left( \left( t, X \right) , \Phi \left( t , X \right) ,  j_{X, \phi_{t}\left(X \right)}^{1} \phi_{t} \right)$ where
$$\phi\left( t,  X \right) =\phi_{t} \left( X \right) = \left( pr_{2}\circ \Phi\right) \left( t , X \right), \ \forall t \in \mathbb{R}, \ \forall X \in \mathcal{B}  $$
Another, less intuitive but easier way to represent an element of $\Phi \left(\mathcal{V}\right)$, is a triple $\left(t,s , j_{X,Y}^{1} \phi\right)$ with $s,t \in \mathbb{R}$, $X \in \mathcal{B}$ and $\phi$ a local automorphism on $\mathcal{B}$ from $X$ to $Y$. Then, the local coordinates of $\Phi \left( \mathcal{V}\right)$ (see Eq. (\ref{15})) are given by

\begin{equation}\label{17.second2313}
\Phi \left(  \mathcal{V}_{\mathcal{U}, \mathcal{W}}\right) : \left(t,s,x^{i} , y^{j}, y^{j}_{i}\right),
\end{equation}
where, for each $ \left( t , s , j_{X,Y}^{1}\phi \right) \in \Phi \left(  \mathcal{V}_{\mathcal{U}, \mathcal{W}}\right) $
\begin{itemize}
\item $t \left( t , s , j_{X,Y}^{1}\phi \right) = t$.
\item $s\left( t , s , j_{X,Y}^{1}\phi \right) = s.$
\item $x^{i} \left( t , s , j_{X,Y}^{1}\phi \right) = x^{i} \left(X\right)$.
\item $y^{j} \left( t , s , j_{X,Y}^{1}\phi \right) = y^{j} \left( Y\right)$.
\item $y^{j}_{i}\left( t , s , j_{X,Y}^{1}\phi \right)  = \dfrac{\partial \left(y^{j}\circ \phi\right)}{\partial x^{i}_{| X} }$.
\end{itemize}
where $\left(x^{i}\right)$ and $\left( y^{i}\right)$ are local charts defined on the open subsets of $\mathcal{B}$, $\mathcal{U}$ and $\mathcal{W}$ respectively, and $\Phi \left(  \mathcal{V}_{\mathcal{U}, \mathcal{W}} \right)$ is given by the triples $\left( t , s , j_{X,Y}^{1}\phi \right)$ such that $X \in \mathcal{U}$ and $Y \in \mathcal{W}$.\\\\
\noindent{Notice that, the space of (local) embeddings $\Phi: \mathcal{C} \rightarrow \mathcal{C}$ of fibre bundles over the identity on $\mathbb{R}$ is easily identified with the set of (local) histories by using the reference configuration. So, \textit{the groupoid $\Phi \left(\mathcal{V}\right) \rightrightarrows \mathcal{C}$ encompasses all the possible histories of the material evolution}. Then, by using Eq. (\ref{1.4.5.5}), we may define $W$ on the space $\Phi \left( \mathcal{V} \right)$,}
$$W : \Phi \left(\mathcal{V}\right) \rightarrow V,$$
as follows,
$$ W \left( t , s , j_{X,Y}^{1} \phi\right) = W\left( t , X , \Phi \right),$$
such that
$$ \Phi \left( s , Y \right) = \left( s , \phi_{0} \circ \phi \left( Y \right) \right), \ \forall \left( s , Y \right) \in \mathcal{C},$$
where $\phi_{0}$ is the reference configuration. Then, $W$ does not depend on the final point, i.e., for all $\left(t, X\right) ,\left( s, Y \right),\left(r , Z \right)\in \mathcal{C}$
\begin{equation}\label{92.second}
 W \left(t,s , j_{X,Y}^{1} \phi\right)  =  W \left(t,r , j_{X,Z}^{1} \left(  \phi_{0}^{-1}\circ \tau_{Z-Y} \circ \phi_{0} \circ \phi\right)\right), 
\end{equation}
for all $\left(t,s , j_{X,Y}^{1} \phi\right) \in \Phi \left(\mathcal{V}\right)$ where $\tau_{v}$ is the translation map on $\mathbb{R}^{3}$ by the vector $v$. This point of view will be useful for our purpose.\\

\begin{definition}
\rm
The \textit{material groupoid} of a body-time manifold with mechanical response $W$ is defined as the largest subgroupoid $\Omega \left( \mathcal{C} \right)\rightrightarrows \mathcal{C}$ of $\Phi \left(\mathcal{V}\right)$ such that leaves $W$ invariant. More explicitly, an element of $\Phi \left(\mathcal{V}\right)$  $\left(t,s , j_{X,Y}^{1} \phi\right)$ is in the material groupoid if and only if
$$ W \left(t,r , j^{1}_{X,Z}\left( \psi \cdot \phi \right) \right) = W \left(s,r , j^{1}_{Y,Z}\psi \right) , $$
for all $\left(s,r , j^{1}_{Y,Z}\psi \right) \in \Phi \left(\mathcal{V}\right)$.
\end{definition}
\noindent{In other words, $\Omega \left( \mathcal{C} \right)$ is the space of all (time-)material isomorphisms (see Definition \ref{1.33.2}). This groupoid was first presented in \cite{MEPMDLSEG}.}\\
The isotropy group at each $\left( t, X \right) \in \mathcal{C}$ will be denoted by $G \left( t , X \right)$ and its elements are the material symmetries at $\left( t,X \right)$. Observe that, as in the spatial case, the resulting groupoid does not have to be a Lie subgroupoid of $\Phi \left(\mathcal{V}\right) \rightrightarrows \mathcal{C}$.
\begin{definition}
\rm
We will also define the \textit{$\left(t,s\right)-$material groupoid} $\Omega_{t,s} \left( \mathcal{B} \right)$ as the set of all material isomorphisms from the instant $t$ to the instant $s$. 
\end{definition}
\noindent{Notice that, when $t=s$, the $\left(t,t\right)-$material groupoid $\Omega_{t,t} \left( \mathcal{B} \right)$ is a subgroupoid of the material groupoid $\Omega \left( \mathcal{C} \right)$. For each instant $t$, $\Omega_{t,t} \left( \mathcal{B} \right)$ is called \textit{$t-$material groupoid} and denoted by $\Omega_{t} \left( \mathcal{B} \right)$}.\\
On the other hand, $\Omega_{t} \left( \mathcal{B} \right)$ may be consider as a subgroupoid of $\Pi^{1} \left( \mathcal{B}, \mathcal{B} \right)$, where we are identifying $\mathcal{B}$ with $\{t\} \times \mathcal{B}$. Notice that, indeed, \textit{$\Omega_{t} \left( \mathcal{B} \right)$ is the material groupoid associated to simple body structure of the state $t$ of the body $\mathcal{B}$, i.e., with this identification,
$$\Omega_{t} \left( \mathcal{B} \right)  = \Omega \left( \mathcal{B}_{t} \right).$$}
\noindent{We will use both interpretations of $\Omega_{t} \left( \mathcal{B} \right)$ indistinctly along the paper.}\\\\
\noindent{As a transversal construction, we will define the \textit{$\left(X,Y\right)-$material groupoid} $\Omega_{X,Y} \left( \mathbb{R} \right)$.
\begin{definition}\label{definition17}
\rm
The \textit{$\left(X,Y\right)-$material groupoid} $\Omega_{X,Y} \left( \mathbb{R} \right)$ is defined as the set of all material isomorphisms from the particle $X$ to the particle $Y$ varying the time variable. 
\end{definition}
\noindent{Again, we may notice that, when $X=Y$, the $\left(X,X\right)-$material groupoid $\Omega_{X,X} \left( \mathbb{R} \right)$ is a subgroupoid of the material groupoid $\Omega \left( \mathcal{C} \right)$. For each material point $X$, $\Omega_{X,X} \left( \mathbb{R} \right)$ is called \textit{$X-$material groupoid} and denoted by $\Omega_{X} \left( \mathbb{R} \right)$.}\\
On the other hand, $\Omega_{X} \left( \mathbb{R} \right)$ may be consider as a subgroupoid of $\left(\mathbb{R}  \times \mathbb{R}\right) \times \Pi^{1} \left( \mathcal{B} , \mathcal{B}\right)_{X}^{X} \rightrightarrows \mathbb{R}$, where we are identifying $\mathbb{R}$ with $\mathbb{R} \times \{X\}$. Furthermore, the structure of Lie groupoid of $\left(\mathbb{R}  \times \mathbb{R}\right) \times \Pi^{1} \left( \mathcal{B} , \mathcal{B}\right)_{X}^{X}$ is given by,
$$\left( s,t,j_{X,X}^{1}\phi\right) \cdot \left( r,s,j_{X,X}^{1}\psi\right) = \left( r,t,j_{X,X}^{1}\left(\phi \circ \psi \right)\right), $$
for all $\left( s,t,j_{X,X}^{1}\phi\right) , \left( r,s,j_{X,X}^{1}\psi\right) \in  \mathbb{R}\times \mathbb{R} \times\Pi^{1} \left( \mathcal{B} , \mathcal{B}\right)_{X}^{X}$. Again, we will use both interpretations of $\Omega_{X} \left( \mathbb{R} \right)$ along the paper.}\\
Roughly speaking, the material groupoid $\Omega \left( \mathcal{C} \right)\rightrightarrows \mathcal{C}$ encompasses the global evolution of the body, the $t-$material groupoid $\Omega_{t} \left( \mathcal{B} \right)$ encodes all the material properties of the body at the instant $t$ and the $X-$material groupoid $\Omega_{X} \left( \mathbb{R} \right)$ embraces all the evolution of the particle $X$.\\\\

\begin{proposition}\label{auxiliarprop34re}
Let $\Omega \left( \mathcal{C} \right)$ be the material groupoid. If $\Omega \left( \mathcal{C} \right)$ is a Lie subgroupoid of $\Phi \left( \mathcal{V} \right)$, then for all instant $t$ and all material point $X$ we have that $\Omega_{t} \left( \mathcal{C} \right)$ and $\Omega_{X} \left( \mathbb{R} \right)$ are Lie subgroupoids of $\Phi \left( \mathcal{V} \right)$.
\begin{proof}
Assume that $\Omega \left( \mathcal{C} \right)$ is a Lie subgroupoid of $\Phi \left( \mathcal{V} \right)$. Let us consider the following submersions,
$$\pi_{1}: \Omega \left( \mathcal{C} \right) \rightarrow \mathcal{B} \times \mathcal{B}, \ \ \ \pi_{2}: \Omega \left( \mathcal{C} \right) \rightarrow \mathbb{R} \times \mathbb{R},$$
given by
$$ \pi_{1}\left( t, s , j_{X,Y}^{1}\phi \right) = \left( X , Y \right), \ \ \ \pi_{2}\left( t, s , j_{X,Y}^{1}\phi \right) = \left( t,s \right),$$
for all $\left(t, s , j_{X,Y}^{1}\phi \right) \in \Omega \left( \mathcal{C} \right)$. Then
$$ \Omega_{X} \left( \mathbb{R} \right) = \pi_{1}^{-1}\left( X, X \right),  \ \ \ \Omega_{t} \left( \mathcal{C} \right) = \pi_{2}^{-1}\left(t,t\right).$$

\end{proof}
\end{proposition}
So, the condition of \textit{``being a Lie groupoid''} is stronger over the material groupoid than over the $t-$material groupoids and $X-$material groupoids.\\\\

\noindent{Notice that, all the defined canonical groupoids satisfy the following short sequences of contents,} 
$$ \Omega_{t} \left( \mathcal{B} \right) \leq \Omega \left( \mathcal{C} \right)  \leq \Phi \left( \mathcal{V} \right), \ \forall t.$$
$$\  \Omega_{X} \left( \mathbb{R} \right) \leq \Omega \left( \mathcal{C} \right)  \leq \Phi \left( \mathcal{V} \right), \ \forall X.$$

\vspace{0.7cm}

Then, we may construct the correspondent characteristic distributions. We will start with the associated characteristic distribution $A \Omega\left( \mathcal{C} \right)^{T}$ to the material groupoid, which will be called \textit{material distribution} of the body-time manifold $\mathcal{C}$. So, in a similar way to the material distribution associated to a \textit{spatial} body $\mathcal{B}$ (see Eq. (\ref{14.22})), $A \Omega\left( \mathcal{C} \right)^{T}$ is generated by the (left-invariant) vector fields on $\Phi \left( \mathcal{V} \right)$ which are in the kernel of $TW$. Equivalently, the material distribution of $\mathcal{C}$ is generated by the left-invariant vector fields $\Theta$ on $\Phi \left( \mathcal{V} \right)$ such that
\begin{equation}\label{matgrou2}
    T W \left( \Theta \right) = 0
\end{equation}
So, let $\Theta$ be a left-invariant vector field on $\Phi \left( \mathcal{V} \right)$. Then,
\begin{equation}
    \Theta  \left(t,s,x^{i} , y^{j}, y^{i}_{j}\right)   = \lambda \dfrac{\partial}{\partial  t}   +   \Theta^{i}\dfrac{\partial}{\partial  x^{i} } + y^{i}_{l}\Theta^{l}_{j}\dfrac{\partial}{\partial y^{i}_{j} }
\end{equation}
respect to a local system of coordinates $\left(t,s,x^{i} , y^{j}, y^{i}_{j}\right)$ on $ \Phi \left(  \mathcal{V}_{\mathcal{U}, \mathcal{V}}\right)$ with $\mathcal{U}$ and $\mathcal{V}$ two open subsets of $\mathcal{B}$ and $ \Phi \left(\mathcal{V}_{\mathcal{U}, \mathcal{V}} \right)$ is given by the triples $\left( t, s , j_{X,Y}^{1}\phi    \right)$ in $\Phi \left( \mathcal{V} \right)$ such that $X \in \mathcal{U}$ and $Y \in \mathcal{V}$. Then, $\Theta$ is an admissible vector field for the couple $\left( \Phi \left( \mathcal{V} \right) ,\Omega \left( \mathcal{C} \right)\right)$ if, and only if,the following equations holds,
\begin{equation}\label{Eqmaterialgroupoid12}
    \lambda \dfrac{\partial W}{\partial  t}  +  \Theta^{i}\dfrac{\partial W}{\partial  x^{i} } + y^{i}_{l}\Theta^{l}_{j}\dfrac{\partial W}{\partial  y^{i}_{j} } = 0.
\end{equation}
Notice that, here $\lambda$, $\Theta^{i}$ and $\Theta^{i}_{j}$ are functions depending on $t$ and $X$. Thus, construct the material distribution is reduced to solve Eq. (\ref{Eqmaterialgroupoid12}). The base-characteristic distribution $A \Omega \left( \mathcal{C} \right)^{\sharp}$ will be called \textit{body-material distribution} and the transitive distribution $A \Omega \left( \mathcal{B} \right)^{B}$ will be called \textit{uniform-material distribution}.\\
The foliations associated to the material distribution, the body-material distribution and uniform-material distribution will be called \textit{material foliation}, \textit{body-material foliation} and \textit{uniform-material foliation} and they will be denoted by $\overline{\mathcal{F}}$, $\mathcal{F}$ and $\mathcal{G}$, respectively.\\\\

\noindent{On the other hand, analogously, for an instant $t$, consider the $t-$material groupoids $\Omega_{t} \left( \mathcal{B} \right)$ as the material groupoid of the state $t$ of the material body $\mathcal{B}$. Therefore, (see Eq. (\ref{14.22})) the associated characteristic distribution $A \Omega_{t} \left( \mathcal{B} \right)^{T}$ to the $t-$material groupoid, which will be called \textit{$t-$material distribution}, is generated by the (left-invariant) vector fields on $\Pi^{1} \left( \mathcal{B} , \mathcal{B}\right)$ which are in the kernel of $TW_{t}$, where $W_{t}$ is given by
$$ W_{t}: \Pi^{1} \left( \mathcal{B} , \mathcal{B}\right) \rightarrow V,$$
such that $W_{t}\left( j_{X,Y}^{1}\phi \right) = W \left( t , t, j_{X,Y}^{1}\phi \right)$ for all $ j_{X,Y}^{1}\phi  \in \Pi^{1} \left( \mathcal{B} , \mathcal{B}\right)$. In other words, the $t-$material distribution of $\mathcal{C}$ is generated by the left-invariant vector fields $\Theta$ on $\Pi^{1} \left( \mathcal{B} , \mathcal{B}\right)$ such that
\begin{equation}\label{tmatgrou1}
    T W_{t} \left( \Theta \right) = 0
\end{equation}
Indeed, it satisfies that $A \Omega_{t} \left( \mathcal{B} \right)^{T}$ is the material distribution of the state $t$ of the material body.\\
So, let $\Theta$ be a left-invariant vector field on $\Pi^{1} \left( \mathcal{B}, \mathcal{B} \right)$. Then,

\begin{equation}
    \Theta  \left(x^{i} , y^{j}, y^{i}_{j}\right)   =   \Theta^{i}\dfrac{\partial}{\partial  x^{i} } + y^{i}_{l} \Theta^{l}_{j}\dfrac{\partial}{\partial y^{i}_{j} },
\end{equation}
respect to a local system of coordinates $\left(x^{i} , y^{j}, y^{i}_{j}\right)$ on $ \Pi^{1} \left(  \mathcal{U} , \mathcal{U}\right)$ (see Eq. (\ref{17})) with $\mathcal{U}$ an open subset of $\mathcal{B}$. Then, $\Theta$ is an admissible vector field for the couple $\left( \Pi^{1}\left( \mathcal{B} , \mathcal{B} \right) , \Omega_{t} \left( \mathcal{B} \right)\right)$ if, and only if, the following equations hold,
\begin{equation}\label{tmatgroup3234}
    \Theta^{i}\dfrac{\partial W_{t}}{\partial  x^{i} } + y^{i}_{l}\Theta^{l}_{j}\dfrac{\partial   W_{t}}{\partial  y^{i}_{j} } = 0
\end{equation}
Note that, here $\Theta^{i}$ and $\Theta^{i}_{j}$ are functions depending on $X$.\\
On the other hand, let us observe that, taking into account the \textit{consistency proposition} \ref{consistencyproperty2904}, as a subgroupoid of $\Phi \left( \mathcal{V} \right)$, the groupoid $\Omega_{t}\left( \mathcal{B} \right)$ is generated by the left-invariant vector fields $\Theta$ on $\Phi \left( \mathcal{V}  \right)$,
\begin{equation}
    \Theta  \left(t,s,x^{i} , y^{j}, y^{i}_{j}\right)   = \lambda \dfrac{\partial}{\partial  t}   +   \Theta^{i}\dfrac{\partial}{\partial  x^{i} } + y^{i}_{l} \Theta^{l}_{j}\dfrac{\partial}{\partial y^{i}_{j} }
\end{equation}
such that $\lambda_{\vert \{t\}\times \{t\} \times \Pi^{1} \left( \mathcal{B} , \mathcal{B}\right)} \equiv 0$ and
\begin{equation}\label{tmatgroup3234.2}
    \Theta^{i}\dfrac{\partial W_{t}}{\partial  x^{i} } + y^{i}_{l}\Theta^{l}_{j}\dfrac{\partial   W_{t}}{\partial  y^{i}_{j} } = 0
\end{equation}
on any material point at the instant $t$. The base-characteristic distribution $A \Omega_{t} \left( \mathcal{B} \right)^{\sharp}$ (see Theorem \ref{10.24}) will be called \textit{$t-$body-material distribution} and the transitive distribution $A \Omega_{t} \left( \mathcal{B} \right)^{B}$ (see Corollary \ref{10.42}) will be called \textit{$t-$uniform-material distribution}.\\
The foliations associated to the $t-$material distribution, the $t-$body-material distribution and $t-$uniform-material distribution will be called \textit{$t-$material foliation}, \textit{$t-$body-material foliation} and \textit{$t-$uniform-material foliation} and they will be denoted by $\overline{\mathcal{F}}_{t}$, $\mathcal{F}_{t}$ and $\mathcal{G}_{t}$, respectively.}\\\\

\noindent{The characteristic distribution associated to the $X-$material groupoid $A \Omega_{X} \left(\mathbb{R} \right)^{T}$ is called \textit{$X-$material distribution}. Analogously, $A \Omega_{X} \left(\mathbb{R} \right)^{T}$ is generated by the (left-invariant) vector fields on $\left(\mathbb{R}  \times \mathbb{R}\right) \times \Pi^{1} \left( \mathcal{B} , \mathcal{B}\right)_{X}^{X}$ which are in the kernel of $TW_{X}$, where $W_{X}$ is given by the restriction of $W$ to $\mathbb{R}\times \mathbb{R} \times\Pi^{1} \left( \mathcal{B} , \mathcal{B}\right)_{X}^{X}$,
$$ W_{X}: \mathbb{R}\times \mathbb{R} \times\Pi^{1} \left( \mathcal{B} , \mathcal{B}\right)_{X}^{X} \rightarrow V$$
In other words, the $X-$material distribution of $\mathcal{C}$ is generated by the left-invariant vector fields $\Theta$ on $\mathbb{R}\times \mathbb{R} \times\Pi^{1} \left( \mathcal{B} , \mathcal{B}\right)_{X}^{X}$ such that
\begin{equation}\label{tmatgrou1.second4}
    T W_{X} \left( \Theta \right) = 0
\end{equation}
Notice that, the groupoid structure of $\left(\mathbb{R}  \times \mathbb{R}\right) \times \Pi^{1} \left( \mathcal{B} , \mathcal{B}\right)_{X}^{X}$ is the unique groupoid structure such that it is a subgroupoid of $\Phi \left( \mathcal{V}\right)$, i.e.,
$$\left( s,t,j_{X,X}^{1}\phi\right) \cdot \left( r,s,j_{X,X}^{1}\psi\right) = \left( r,t,j_{X,X}^{1}\left(\phi \circ \psi \right)\right), $$
for all $\left( s,t,j_{X,X}^{1}\phi\right) , \left( r,s,j_{X,X}^{1}\psi\right) \in  \mathbb{R}\times \mathbb{R} \times\Pi^{1} \left( \mathcal{B} , \mathcal{B}\right)_{X}^{X}$.\\
So, let $\Theta$ be a left-invariant vector field on $\mathbb{R}\times \mathbb{R} \times \Pi^{1} \left( \mathcal{B} , \mathcal{B}\right)_{X}^{X}$. Then,

\begin{equation}
    \Theta  \left(t,s, y^{i}_{j}\right)   =   \lambda\dfrac{\partial}{\partial  t } +  y^{i}_{l} \Theta^{l}_{j}\dfrac{\partial}{\partial y^{i}_{j} },
\end{equation}
respect to a local system of coordinates $\left(t,s, y^{i}_{j}\right)$ on $ \mathbb{R}\times \mathbb{R} \times \Pi^{1} \left( \mathcal{U} , \mathcal{U}\right)_{X}^{X}$ with $\mathcal{U}$ an open subset of $\mathcal{B}$ with $X \in \mathcal{U}$. Then, $\Theta$ is an admissible vector field for the couple $\left( \Phi \left( \mathcal{V} \right) , \Omega_{X} \left( \mathbb{R} \right)\right)$ if, and only if, the following equations hold,
\begin{equation}\label{Xmatgroup323456}
    \lambda\dfrac{\partial W_{X}}{\partial  t } + y^{i}_{l} \Theta^{l}_{j}\dfrac{\partial   W_{X}}{\partial  y^{i}_{j} } = 0
\end{equation}
Observe that, here $\lambda$ and $\Theta^{i}_{j}$ are functions depending on $t$.\\
On the other hand, taking into account the \textit{consistency proposition} \ref{consistencyproperty2904}, as a subgroupoid of $\Phi \left( \mathcal{V} \right)$, the groupoid $\Omega_{X}\left( \mathbb{R} \right)$ is generated by the left-invariant vector fields $\Theta$ on $\Phi \left( \mathcal{V}  \right)$,
\begin{equation}
    \Theta  \left(t,s,x^{i} , y^{j}, y^{i}_{j}\right)   = \lambda \dfrac{\partial}{\partial  t}   +   \Theta^{i}\dfrac{\partial}{\partial  x^{i} } + y^{i}_{l} \Theta^{l}_{j}\dfrac{\partial}{\partial y^{i}_{j} }
\end{equation}
such that $\Theta^{i}_{\vert \mathbb{R}\times \mathbb{R} \times \Pi^{1} \left( \mathcal{U} , \mathcal{U}\right)_{X}^{X}  } \equiv 0$ and
\begin{equation}\label{tmatgroup3234.3}
    \Theta^{i}\dfrac{\partial W_{X}}{\partial  x^{i} } + y^{i}_{l} \Theta^{l}_{j}\dfrac{\partial   W_{X}}{\partial  y^{i}_{j} } = 0,
\end{equation}
at any instant for a fixed material point $X$. The base-characteristic distribution $A \Omega_{X} \left( \mathbb{R}\right)^{\sharp}$ (see Theorem \ref{10.24}) will be called \textit{$X-$body-material distribution} and the transitive distribution $A \Omega_{X} \left( \mathbb{R}  \right)^{B}$ (see Corollary \ref{10.42}) will be called \textit{$X-$uniform-material distribution}.\\
The foliations associated to the $X-$material distribution, the $X-$body-material distribution and $X-$uniform-material distribution will be called \textit{$X-$material foliation}, \textit{$X-$body-material foliation} and \textit{$X-$uniform-material foliation} and they will be denoted by $\overline{\mathcal{F}}_{X}$, $\mathcal{F}_{X}$ and $\mathcal{G}_{X}$, respectively. It is important do not confuse $\overline{\mathcal{F}}_{X}$ (resp. $\mathcal{F}_{X}$ and $\mathcal{G}_{X}$), the $X-$material foliation (resp. $X-$body-material foliation and $X-$uniform-material foliation), with $\overline{\mathcal{F}}\left(\epsilon\left( X\right)\right)$ (resp. $\mathcal{F}\left( X\right)$ and $\mathcal{G}\left( X \right)$), the leaf at $\epsilon\left( X\right)$ (resp. the leaf at $X$) of the foliation $\overline{\mathcal{F}}$ (res. $\mathcal{F}$ and $\mathcal{G}$).}\\

To summarize, around an evolution material $\mathcal{C}$, we have constructed the following canonical short sequences of groupoids

$$ \Omega_{t} \left( \mathcal{B} \right) \leq \Omega \left( \mathcal{C} \right) \leq  \Phi \left( \mathcal{V} \right), \ \forall t.$$
$$ \Omega_{X} \left( \mathbb{R} \right) \leq \Omega \left( \mathcal{C} \right) \leq  \Phi \left( \mathcal{V} \right), \ \forall t.$$
and the following canonical short sequences of distributions
\begin{itemize}
    \item[] \hspace{3cm} $A\Omega_{t} \left( \mathcal{B} \right)^{T}  \leq  A\Omega \left( \mathcal{C} \right)^{T}\leq   T\Phi \left( \mathcal{V} \right),  \forall t.$
     \item[] \hspace{3cm}$A\Omega_{X} \left( \mathbb{R} \right)^{T}  \leq  A\Omega \left( \mathcal{C} \right)^{T}\leq   T\Phi \left( \mathcal{V} \right),  \forall t.$
    \item[] \hspace{3cm}$A\Omega_{t} \left( \mathcal{B} \right)^{B} \leq  A\Omega \left( \mathcal{C} \right)^{B}  \leq   T\mathcal{C},  \forall t.$
    \item[] \hspace{3cm}$A\Omega_{X} \left( \mathbb{R} \right)^{B} \leq  A\Omega \left( \mathcal{C} \right)^{B}  \leq    T\mathcal{C},  \forall t.$
    \item[] \hspace{3cm}$A\Omega_{t} \left( \mathcal{B} \right)^{\sharp} \  \leq  A\Omega \left( \mathcal{C} \right)^{\sharp} \ \leq    T \mathcal{C},  \forall t.$
    \item[] \hspace{3cm}$A\Omega_{X} \left( \mathbb{R} \right)^{\sharp} \  \leq  A\Omega \left( \mathcal{C} \right)^{\sharp} \ \leq    T \mathcal{C},  \forall t.$
\end{itemize}

\part{Remodeling}

As opposed to the uniformity in the spatial case, arise new material properties associated to the evolution of the body. In particular, the temporal counterpart of uniformity is a specific case of evolution of the material called \textit{remodeling}. This part will be focused on the study of \textit{global remodeling}, as one of the main contributions of this paper.

\begin{definition}\label{1.17.2.se}
\rm
Let $\mathcal{C}$ be a body-time manifold:
\begin{itemize}
    \item A material particle $X \in \mathcal{B}$ is presenting a \textit{remodeling} when it is connected with all the instants by a material isomorphism, i.e., all the points at $\mathbb{R} \times \{X\}$ are connected by material isomorphisms.
    \item  $\mathcal{C}$ is presenting a \textit{global remodeling} or simply a \textit{remodeling} when all the material points are presenting a remodeling.
    \item  We will say that $\mathcal{C}$ is presenting a \textit{uniform remodeling} when it is presenting a remodeling and some (and hence all) state is uniform.
    \item \textit{Growth} and \textit{resorption} are given by a remodeling with volume increase or volume decrease of the material body $\mathcal{B}$. 
\end{itemize}

\end{definition}
Intuitively, a material evolution presents a remodeling when the constitutives properties of the material does not change with the time. This kind evolution may be found in biological tissues \cite{RODRIGUEZ1994455}. Wolff's law of trabecular architecture of bones (see for instance \cite{TURNER19921}) is a relevant example. Here, trabeculae are assumed to change their orientation following the principal direction of stress. It is important to note that the fact of that the material body remains materially isomorphic with the time does not preclude the possibility of adding (growth) or removing (resorption) material, as long as the material added is \textit{of the same type}.
\begin{proposition}
Let $\mathcal{C}$ be a body-time manifold. A material particle $X \in \mathcal{B}$ is presenting a remodeling if, and only if, the $X-$material groupoid $\Omega_{X}\left(  \mathbb{R}\right)$ is transitive. $\mathcal{C}$ is presenting a remodeling if, and only if, for all material point $X$, the $X-$material groupoid $\Omega_{X}\left(  \mathbb{R}\right)$ is transitive.
\end{proposition}
\begin{corollary}\label{corollaryunifremod2344}
Let $\mathcal{C}$ be a body-time manifold. the material groupoid $\Omega \left( \mathcal{C}\right)$ is transitive if, and only if, $\mathcal{C}$ is presenting a uniform remodeling.
\end{corollary}

Observe that, analogously to uniformity, the definition of remodeling is pointwise. Consider a material particle $X_{0}$ which presents a remodeling or, equivalently, there exists a map 

\begin{equation}\label{1.564.34.5.second}
    P: \mathbb{R} \rightarrow Gl \left( 3, \mathbb{R} \right)
\end{equation}
such that, for all $t \in \mathbb{R}$, $P \left(t\right)$ is a material isomorphism from $\left( t_{0} , X_{0}\right)$ to $\left( t , X_{0}\right)$ for a fixed time $t_{0}$. Nevertheless, the differentiability condition of $P$ is not guaranteed.
\begin{definition}\label{1.7.2.second}
\rm
Let be a body-time manifold $\mathcal{C}$. A material point $X_{0}$ is said to be presenting a \textit{smooth remodeling} if for each point $t \in \mathbb{R}$ there is an interval $I $ around $t$ and a smooth map $P : I \rightarrow Gl \left( 3, \mathbb{R} \right)$ such that for all $s \in I$ it satisfies that $P \left(s\right)$ is a material isomorphism from $\left(t,X_{0} \right)$ to $\left(s,X_{0} \right)$. The map $P$ is called a \textit{right (local) smooth remodeling process at} $X_{0}$. A \textit{left (local) smooth remodeling process at} $X_{0}$ is defined in a similar way.
\end{definition}

\subsubsection*{Mass consistency condition}
Notice that the definition of material isomorphism does not include any relation to the mass density of the body. However, it is desirable to impose some kind of condition to be consistent with the mass density.\\
Thus, for each instant of time, a volume form is specified, i.e., we have $\omega \left( t \right)$ a time dependent volume form on $\mathcal{B}$. Let $X_{0}$ be a material particle presenting a remodeling. Without loss of generality, we assume that the remodelling process $P$ satisfies the initial condition, 
$$P \left( 0 \right)  = I.$$
Then, \textit{mass consistency condition} (\cite{EPSBOOK2,EPSTEIN201572}) consists of the imposition on the remodeling process at $X_{0}$ of that it preserves the volume form. In other words, a (local) right smooth remodeling process $P$ at $X_{0}$ satisfies the mass consistency condition, if and only if,
$$P\left( t\right)^{*} \omega \left( t \right) = \omega \left( 0 \right), \ \forall t \in I.$$
Then, equivalently, associated mass density, $\rho \left( t \right) = \vert \omega \left( t \right) \vert$, should satisfy that
\begin{equation}\label{232.second}
    \rho \left( t \right) = \vert J_{P\left(t\right)} \vert^{-1} \rho \left(0\right)
\end{equation}
where $J_{P\left( t \right)}$ is the determinant of $P\left( t \right)$.
We will also assume that $P$ is orientation-preserving, i.e., $J_{P\left( t \right)}>0$.\\
Calculating the time derivatives of Eq. (\ref{232.second}),

\begin{eqnarray*}
\dot{\rho}\left( t \right)  & = &  \dot{\left( J_{P\left(t\right)}^{-1} \rho \left(0\right)\right)}\\
  & = &- \rho \left(0\right) J^{-2}_{P\left(t\right)}\left[ \dot{J}_{P\left( t \right)} \right]\\
  & = & -\rho \left(0\right) J^{-2}_{P\left(t\right)}\left[  J_{P\left( t \right)} Tr \left(P^{-1}\left(t\right) \cdot \dot{P}\left(t\right)\right) \right]\\
  & = & -\rho \left(0\right) J^{-1}_{P\left(t\right)}  Tr \left(P^{-1}\left(t\right) \cdot \dot{P}\left(t\right)\right) \\
 & = & -\rho \left(t\right) Tr \left(P^{-1}\left(t\right) \cdot \dot{P}\left(t\right)\right)
\end{eqnarray*}
The term $L_{P\left(t\right)} = P^{-1}\left(t\right) \cdot \dot{P}\left(t\right)$ is called \textit{remodeling velocity gradient}.
\begin{proposition}
Let $\mathcal{C}$ be a body-time manifold and $X_{0}$ be a material particle. A remodeling process $P$ is producing growth if, and only if, the trace of the remodeling velocity gradient is negative. Conversely, resorption is equivalent to a positive trace of the remodeling velocity gradient.
\begin{proof}
The trace of the remodeling velocity gradient $L_{P\left(t\right)} = P^{-1}\left(t\right) \cdot \dot{P}\left(t\right)$ is negative (resp. positive) if, and only if, $\rho $ is an increasing (resp. decreasing) function or, in other words, the volume of $\mathcal{B}$ respect to $\omega \left(t \right)$ is increasing (resp. decreasing).
\end{proof}
\end{proposition}
Several interesting examples of remodeling processes may be found in the literature. In particular, in \cite{EPSBOOK2} it is used a model for orthotropic solids in which the tensor $P$ is proper orthogonal at all times. This model simulates an evolution law in trabeculae bones.\\

\noindent{Let us assume that $\mathcal{B}$ is uniform. Then, $\mathcal{B}$ is uniform in all its states. So, for a fixed point $\left(t_{0},X_{0} \right) \in \mathcal{C}$ we may find a map} 

\begin{equation}\label{1.564.34.5}
    P: \mathcal{C} \rightarrow Gl \left( 3, \mathbb{R} \right)
\end{equation}
such that, for all $\left(t,Y\right) \in \mathcal{C}$, $P \left(t , Y \right)$ is a material isomorphism from $\left( t_{0} , X_{0}\right)$ to $\left( t , Y\right)$. \textit{However, even when all the particles present smooth remodeling, $P$ does not have to be differentiable}. In other words, roughly speaking, the evolution of all the particles along the time could be ``\textit{smooth}'', but the change from the time-evolution of one particle to another could still be ``\textit{abrupt}'' (not differentiable).\\
Thus, \textit{we cannot define \textit{smooth remodeling} over the whole material evolution as the smooth remodeling at all points, we still need a more restrictive definition of smoothness on the evolution of the material body.}
\begin{definition}\label{1.7.2}
\rm
A body-time manifold $\mathcal{C}$ with some (and hence all of them) state uniform is said to be presenting a \textit{smooth uniform remodeling} if for each point $\left(t,X \right) \in \mathcal{C}$ there is a neighbourhood $\mathcal{U} $ around $\left(t,X \right)$ and a smooth map $P : \mathcal{U} \rightarrow Gl \left( 3, \mathbb{R} \right)$ such that for all $\left(s,Y \right) \in \mathcal{U}$ it satisfies that $P \left(s,Y\right)$ is a material isomorphism from $\left(t,X \right)$ to $\left(s,Y \right)$. The map $P$ is called a \textit{right (local) smooth field of material isomorphisms}. A \textit{left (local) smooth field of material isomorphisms} is defined analogously.
\end{definition}
One could think that it is reasonable that a non-uniform body present a smooth remodeling. However, the definition of this kind of smooth remodeling (more general) is not clear. One of the contributions of this paper is the use of material distributions to define and characterize this kind of smooth remodeling for non-uniform bodies (Definition \ref{smoothremd2345}).
\begin{proposition}\label{1.18.2.again}
Let $\mathcal{C}$ be a body-time manifold such that some (and hence all of them) state is uniform. Then, $\mathcal{C}$ is presenting a (smooth) uniform remodeling if, and only if, there exist (differentiable) maps $\overline{W}: Gl \left( 3, \mathbb{R} \right) \rightarrow V$ and $P : \mathcal{U} \rightarrow Gl \left( 3, \mathbb{R} \right)$ covering $\mathcal{C}$ satisfying,
\begin{equation}\label{1.8.2}
W \left(s, Y , F \right) = \overline{W}\left( F \cdot P \left(s, Y \right)\right).
\end{equation}

\begin{proof}
The proof of this proposition is analogous to Proposition \ref{1.18}.
\end{proof}
\end{proposition}
\noindent{Let us consider now $W$ as a map on $\Phi \left( \mathcal{V} \right)$.}
\begin{proposition}\label{4.40.second234.Uniform345}
Let be a body-time manifold $\mathcal{C}$ with some (and hence all of them) state uniform. $\mathcal{C}$ is presenting a smooth uniform remodeling if, and only if, for each instant $t$ and each material point $X$ there is an open neighbourhood $\mathcal{D}\subset \mathcal{C}$ around $\left( t , X \right)$ such that for all $\left(s,Y \right) \in \mathcal{D}$ and $\left( s,t , j_{Y,X}^{1} \phi \right) \in \Omega \left( \mathcal{C} \right)$ there exists a local section $\mathcal{P}$ of the source map $\overline{\alpha}$ of $\Omega \left( \mathcal{C}\right)$ to the $\overline{\beta}-$fibre $\Omega\left( \mathcal{C}\right)^{\left(t,X\right)}$,
$$ \overline{\alpha}_{\left(t,X\right)} :\Omega \left( \mathcal{C}\right)^{\left(t,X\right)} \rightarrow \mathcal{C},$$
from $\epsilon \left( t ,X\right)$ to $\left( s,t , j_{Y,X}^{1} \phi \right)$.
\end{proposition}
For these reasons, (local) sections of $\overline{\alpha}_{\left( t , X \right)}$ will be called \textit{left local (smooth) field of material isomorphisms at $\left( t,X \right)$}. On the other hand, local sections of
$$ \overline{\beta}^{\left( t , X \right)} :\Omega \left( \mathcal{C}\right)_{\left( t , X \right)} \rightarrow \mathcal{C},$$
will be called \textit{right local (smooth) fields of material isomorphisms at $\left( t , X \right)$}.\\
Hence, $\mathcal{C}$ is presenting a smooth uniform remodeling if, and only if, for any points $\left(t,X\right)$ and $\left( s, Y \right)$, there are two open neighbourhoods $\mathcal{D}$ and $\mathcal{E}$ respectively and a differentiable map
$$\mathcal{P} : \mathcal{D} \times \mathcal{E}  \rightarrow \Omega \left( \mathcal{C} \right) \subseteq \Phi \left( \mathcal{V}  \right),$$
which is a section of the anchor map $\left( \alpha , \beta \right)$ of $\Phi \left( \mathcal{V}  \right)$. When $t=s$ we may assume $\mathcal{D} = \mathcal{E}$ and $\mathcal{P}$ is a morphism of groupoids over the identity map, i.e.,

$$ \mathcal{P} \left( \left( r, Z\right), \left(l , T \right)\right) = \mathcal{P} \left(   \left( m, S\right), \left(l , T \right)\right) \mathcal{P} \left(  \left( r, Z\right), \left(m , S \right) \right),$$
for all $\left( r, Z\right), \left(l , T \right), \left(m,S\right) \in \mathcal{D}$. These kind of maps are called \textit{local (smooth) field of material isomorphisms.}

\begin{corollary}\label{4.4.second2324.uniform2445}
Let be a body-time manifold $\mathcal{C}$ with some (and hence all of them) state uniform. $\mathcal{C}$ is presenting a smooth uniform remodeling if, and only if, $\Omega \left( \mathcal{C}\right)$ is a transitive Lie subgroupoid of $\Phi \left( \mathcal{V} \right)$.
\begin{proof}
Suppose that $\mathcal{C}$ is presenting a smooth uniform remodeling. Let be a triple $\left( s,t , j_{Y,X}^{1} \phi \right) \in \Omega \left( \mathcal{C} \right)$ and a local (smooth) field of material isomorphism through $\left( s,t , j_{Y,X}^{1} \phi \right)$,
$$\mathcal{P} : \mathcal{D} \times \mathcal{E}  \rightarrow \Omega \left( \mathcal{C} \right) \subseteq \Phi \left( \mathcal{V}  \right) \subseteq \Phi \left( \mathcal{V}  \right).$$
Then, the local structure of manifold is given by the charts $\Psi_{\mathcal{D},\mathcal{E}} : \Omega \left( \mathcal{D},\mathcal{E}\right) \  \rightarrow \  \mathbb{R} \times \mathbb{R} \times \Omega \left( \mathcal{C} \right)_{\left(s,Y\right)}^{\left( t,X\right)}$ such that,
$$\Psi_{\mathcal{D},\mathcal{E}}\left( k,l , j_{Z,T}^{1} \psi \right)=  \left( k, l  , \mathcal{P} \left( \left( l, T \right), \left(t,X \right)\right) \left[ \left( k,l , j_{Z,T}^{1} \psi \right) \right] \mathcal{P} \left( \left(s,Y \right), \left(k,Z \right) \right)\right),$$

\noindent{for all $\left( k,l , j_{Z,T}^{1} \psi \right) \in \Omega \left( \mathcal{D},\mathcal{E}\right)$. Here, $\Omega \left( \mathcal{D},\mathcal{E}\right)$ is the set of material isomorphisms from $ \mathcal{D}$ to instants at $\mathcal{E}$.}
\end{proof}
\end{corollary}
Again, we have here a clear difference between a process of remodeling of a uniform body and a process of \textit{smooth} remodeling of a uniform body (see Corollary \ref{1.17.2.se}).\\
Of course, the existence of fields of material isomorphisms is not canonical. Indeed, for a (local) smooth field of material isomorphisms
$$\mathcal{P} : \mathcal{D} \times \mathcal{D}  \rightarrow \Omega \left( \mathcal{C} \right) \subseteq \Phi \left( \mathcal{V}  \right),$$
any other remodeling process $\mathcal{Q}$ satisfies that
$$\mathcal{Q} \left(\left(s,Y \right), \left(k,Z \right) \right) \in   \mathcal{P} \left( \left(t_{0},X_{0}\right), \left(k,Z \right) \right) \cdot  \Omega \left( \mathcal{C} \right)_{\left(t_{0},X_{0}\right)}^{\left( t_{0},X_{0}\right)} \cdot \mathcal{P} \left( \left(s,Y \right), \left(t_{0},X_{0}\right) \right),$$
for a fixed point $\left( t_{0} , X_{0} \right)$ at $\mathcal{D}$. Thus, the symmetry groups work of $\Omega \left( \mathcal{C} \right) $ as a measure of the degree of freedom available in the choice of the fields of material isomorphisms.\\\\

\noindent{Let $\Omega \left( \mathcal{C} \right)$ be the material groupoid associated to the body-time manifold $\mathcal{C}$. Then, we may consider the material distribution $A \Omega \left( \mathcal{C} \right)$, body-material distribution $A \Omega \left( \mathcal{C} \right)^{\sharp}$ and the uniform-material distribution $A \Omega \left( \mathcal{C} \right)^{B}$ and their associated foliations, the material foliation $\overline{\mathcal{F}}$, body-material foliation $\mathcal{F}$ and uniform-material foliation $\mathcal{G}$, respectively.}\\
\begin{theorem}\label{14.1.second323.globalevolutionsd}
Let be a body-time manifold $\mathcal{C}$. The body-material foliation $\mathcal{F}$ (resp. uniform material foliation $\mathcal{G}$) divides $\mathcal{C}$ into maximal smooth uniform remodeling processes (resp. uniform remodeling processes).
\end{theorem}
Notice that, the foliations $\mathcal{F}$ and $\mathcal{G}$ are foliations of the evolution material $\mathcal{C}$. Hence, each leaf is a submanifold of $\mathcal{C}$, i.e., it defines a material evolution of a material submanifold of $\mathcal{B}$ (see Definition \ref{materialsubevol56}). So, in general, it cannot be properly written as a product space 
\begin{equation}\label{productspace}
    \mathbb{R}\times \mathcal{N},
\end{equation}
with $\mathcal{N}$ a submanifold of $\mathcal{B}$.
Nevertheless, this impossibility turns out to be the most natural (see below of Definition \ref{materialsubevol56}).\\
Notice that, the dimensions of the leaves of the body-material foliation $\mathcal{F}$ (resp. uniform material foliation $\mathcal{G}$) are the dimensions of the fibres of $A\Omega \left( \mathcal{C} \right)^{\sharp}_{\left(t,X\right)}$ (resp. $ A \Omega \left( \mathcal{C} \right)^{B}_{\left(t,X\right)} $). So, may prove the following result:
\begin{theorem}\label{14.1.second323.globalevolutionsd.dimesions}
Let be a body-time manifold $\mathcal{C}$. $\mathcal{C}$ presents a smooth uniform remodeling process (resp. uniform remodeling) if, and only if, $dim  \left( A \Omega \left( \mathcal{C} \right)^{\sharp}_{\left(t,X\right)}\right) = 4 $ (resp. $dim  \left( A \Omega \left( \mathcal{C} \right)^{B}_{\left(t,X\right)}\right) = 4 $) for all instant $t$ and particle $X$, with $A \Omega \left( \mathcal{C}\right)^{\sharp}_{\left(t,X\right)}$ (resp. $A \Omega \left( \mathcal{C}\right)^{B}_{\left(t,X\right)}$) the fibre of $A \Omega \left( \mathcal{C}\right)^{\sharp}$ (resp. $A \Omega \left( \mathcal{C}\right)^{B}$) at $\left(t,X\right)$.
\end{theorem}
Therefore, this theorem bring us a computational condition of testing the property of being a ``\textit{process of remodeling}''. In particular, we will have to study Eq. (\ref{Eqmaterialgroupoid12}),
\begin{equation}\label{Eqmaterialgroupoid12timedependent}
    \lambda \dfrac{\partial W}{\partial  t}  +  \Theta^{i}\dfrac{\partial W}{\partial  x^{i} } + y^{i}_{l}\Theta^{l}_{j}\dfrac{\partial W}{\partial  y^{i}_{j} } = 0.
\end{equation}
where $\lambda$, $\Theta^{i}$ and $\Theta^{i}_{j}$ are functions depending on $t$ and $X$. So, the material evolution is presenting a process of remodeling if we may find $4$ linearly independent solutions to this equation.\\\\

\noindent{For each instant $t$, let us recall the $t-$material distribution $A \Omega_{t} \left( \mathcal{B} \right)^{T}$, its associated $t-$body-material distribution $A \Omega_{t} \left( \mathcal{B} \right)^{\sharp}$ and $t-$uniform-material $A \Omega_{t} \left( \mathcal{B} \right)^{B}$ and associated foliations $t-$material foliation $\overline{\mathcal{F}}_{t}$, $t-$body-material foliation $\mathcal{F}_{t}$ and $t-$uniform-material foliation $\mathcal{G}_{t}$.}\\
We have proved that the $t-$material groupoid $\Omega_{t} \left( \mathcal{B} \right)$ is just the material groupoid of the state $t$ of the body $\mathcal{B}$. Therefore, by using Theorem \ref{14.1} we have that,
\begin{theorem}\label{14.1.second}
The $t-$body-material foliation $\mathcal{F}_{t}$ (resp. $t-$uniform material foliation $\mathcal{G}_{t}$) divides the state $t$ of the body $\mathcal{B}$ into maximal smoothly uniform material submanifolds (resp. uniform material submanifolds).
\end{theorem}
So, at any instant of time $t$, we have the body divided into ``\textit{smoothly uniform parts}'' and we can see how these parts change along time varying $t$.

\begin{proposition}\label{proprelacionimport234}
Let $\mathcal{C}$ be a material evolution and $\left( t , X \right)$ be a point in $\mathcal{C}$. Then, it satisfies that,
\begin{equation}
    \left( \{t\}\times \mathcal{B} \right)\cap  \mathcal{F} \left( t , X \right) = \{t\} \times \mathcal{F}_{t} \left(  X \right)
\end{equation}

\begin{proof}
Notice that, by construction we have that,
$$\{t\} \times \mathcal{F}_{t} \left(  X \right) \subseteq \left( \{t\}\times \mathcal{B} \right)\cap  \mathcal{F} \left( t , X \right).$$
On the other hand, let $\Theta$ be an admissible vector field for the couple $\left(  \Phi \left( \mathcal{V} \right) , \Omega \left( \mathcal{C} \right) \right)$. Then, $\Theta$ should satisfy Eq. (\ref{Eqmaterialgroupoid12}), i.e., 
    \begin{equation}
    \lambda \dfrac{\partial W}{\partial  t}  +  \Theta^{i}\dfrac{\partial W}{\partial  x^{i} } + y^{i}_{l}\Theta^{l}_{j}\dfrac{\partial W}{\partial  y^{i}_{j} } = 0.
\end{equation}
    where,
    \begin{equation}\label{coordinates453}
    \Theta  \left(t,s,x^{i} , y^{j}, y^{i}_{j}\right)   = \lambda \dfrac{\partial}{\partial  t}   +   \Theta^{i}\dfrac{\partial}{\partial  x^{i} } + y^{i}_{l}\Theta^{l}_{j}\dfrac{\partial}{\partial y^{i}_{j} }
\end{equation}
respect to a local system of coordinates $\left(t,s,x^{i} , y^{j}, y^{i}_{j}\right)$ on $ \Phi \left(  \mathcal{V}_{\mathcal{U}}\right)$ with $\mathcal{U}$ an open subset of $\mathcal{B}$ and $\mathcal{V}_{\mathcal{U}}$ given by the triples $\left( t, s , j_{X,Y}^{1}\phi    \right)$ in $\Phi \left( \mathcal{V} \right)$ such that $X,Y \in \mathcal{U}$. Let us consider two cases,
\begin{itemize}
    \item $T_{\left( t , X\right)} \rho \left( \Theta^{\sharp}\left( t , X\right)\right)= 0$, for all projection $\Theta^{\sharp}$ of an admissible vector field $\Theta$ for the couple $\left(  \Phi \left( \mathcal{V} \right) , \Omega \left( \mathcal{C} \right) \right)$.\\\\
    
    So, any admissible vector field $\Theta$ for the couple $\left(  \Phi \left( \mathcal{V} \right) , \Omega \left( \mathcal{C} \right) \right)$ satisfies that $\lambda\left( t,X \right)=0$ is the local expression (\ref{coordinates453}). Hence, 
    it satisfies the equation
    \begin{equation}
    \Theta^{i}\dfrac{\partial W}{\partial  x^{i} } + y^{i}_{l}\Theta^{l}_{j}\dfrac{\partial W}{\partial  y^{i}_{j} } = 0.
    \end{equation}
 
    Therefore, by Eq. (\ref{tmatgroup3234}), $\Theta$ is an admissible vector fields $\Theta$ for the couple $\left(  \Phi \left( \mathcal{V} \right) , \Omega_{t} \left( \mathcal{B} \right) \right)$, i.e.,
    $$ \left( \{t\}\times \mathcal{B} \right)\cap  \mathcal{F} \left( t , X \right) =   \mathcal{F} \left( t , X \right) \subseteq \{t\} \times \mathcal{F}_{t} \left(  X \right). $$

    \item $T_{\left( t , X\right)} \rho \left( \Theta^{\sharp}\left( t , X\right)\right) \neq 0$, for some projection $\Theta^{\sharp}$ of an admissible vector field $\Theta$ for the couple $\left(  \Phi \left( \mathcal{V} \right) , \Omega \left( \mathcal{C} \right) \right)$.\\\\
    
    Then, $T_{\left( t , X\right)} \rho \left( A \Omega \left( \mathcal{C} \right)^{\sharp}\right) = \mathbb{R}$. Thus, we have that
    $$ T_{\left( t , X \right)} \left( \{t\}\times \mathcal{B} \right)  +     T_{\left(t,X\right)} \mathcal{F} \left( t , X \right) =T_{\left( t , X \right)}\mathcal{C},$$
    i.e., $ \left( \{t\}\times \mathcal{B} \right)$ and $  \mathcal{F} \left( t , X \right) $ are transversal submanifolds of $\mathcal{C}$. Therefore, $ \left( \{t\}\times \mathcal{B} \right)\cap  \mathcal{F} \left( t , X \right) $ is a submanifold of $\mathcal{C}$ and
    
    $$ T_{\left( t,X\right)} \left[\left( \{t\}\times \mathcal{B} \right)\cap  \mathcal{F} \left( t , X \right)\right] = T_{\left( t , X\right)} \left( \{t\}\times \mathcal{B} \right) \cap  T_{\left(t,X\right)}\mathcal{F} \left( t , X \right) .$$
    
    Thus, the tangent vector fields to $\left(\{t \}\times \mathcal{B}  \right)\cap  \mathcal{F} \left( t , X \right)$ are the projections $\Theta^{\sharp}$ of admissible vector fields $\Theta$ for the couple $\left(  \Phi \left( \mathcal{V} \right) , \Omega \left( \mathcal{C} \right) \right)$ such that $\Theta^{\sharp}$ projected on $\mathbb{R}$ is zero, i.e.,

    \begin{equation}
    \Theta^{i}\dfrac{\partial W}{\partial  x^{i} } + y^{i}_{l}\Theta^{l}_{j}\dfrac{\partial W}{\partial  y^{i}_{j} } = 0.
    \end{equation}
    where,
    \begin{equation}
    \Theta  \left(t,s,x^{i} , y^{j}, y^{i}_{j}\right)   =    \Theta^{i}\dfrac{\partial}{\partial  x^{i} } + y^{i}_{l}\Theta^{l}_{j}\dfrac{\partial}{\partial y^{i}_{j} }
    \end{equation}
    Then, by Eq. (\ref{tmatgroup3234}), $\Theta$ is an admissible vector fields $\Theta$ for the couple $\left(  \Phi \left( \mathcal{V} \right) , \Omega_{t} \left( \mathcal{B} \right) \right)$.

    \end{itemize}

\end{proof}
\end{proposition}
In other words, in case we freeze an instant of time $s$ in $\mathcal{F} \left( t , X \right)$, we recover the leaf $\{s\} \times \mathcal{F}_{s} \left(  X \right)$. So, if could write $\mathcal{F} \left( t , X \right)$ as in Eq. (\ref{productspace}), we would be precluding the case in which the shape of leaves ${F}_{s} \left(  X \right)$ change with the time, i.e., \textit{each of the leaves $\mathcal{F} \left( t , X \right)$ present a remodeling in which the uniform leaves can change}.\\\\ \noindent{Therefore, if the foliations $\mathcal{F}_{t}$ (resp. $\mathcal{G}_{t}$) permits us to watch how change the smoothly uniform leaves (resp. uniform leaves) of the body with the time, the foliation $\mathcal{F}$ (resp. $\mathcal{G}$) also show us how time is divided optimally in such a way that at each interval the material evolution presents a smooth remodeling (resp. remodeling) process of all the leaves at the same time}.\\\\

\noindent{Finally, we will present a definition of (non-uniform) smooth remodeling is inspired in Corollary \ref{4.4.second2324.uniform2445}.}
\begin{definition}\label{smoothremd2345}
\rm
Let be a body-time manifold $\mathcal{C}$. $\mathcal{C}$ is presenting a \textit{smooth remodeling} if $\Omega \left( \mathcal{C}\right)$ is a Lie subgroupoid of $\Phi \left( \mathcal{V} \right)$ and, for all particle $X$, the $X-$material groupoid $\Omega_{X} \left( \mathbb{R} \right)$ is a transitive Lie subgroupoid of $\Phi \left( \mathcal{V} \right)$

\end{definition}
Notice that, taking into account Proposition \ref{auxiliarprop34re}, if $\Omega \left( \mathcal{C}\right)$ is a Lie subgroupoid of $\Phi \left( \mathcal{V} \right)$, then for all particle $X$, the $X-$material groupoid $\Omega_{X} \left( \mathbb{R} \right)$ is a Lie subgroupoid of $\Phi \left( \mathcal{V} \right)$. So, the unique requirement on the $X-$material groupoids is transitivity.\\
Definition \ref{smoothremd2345} express mathematically the idea of that the material body varies smoothly through the time and the intrinsecal properties does not change. It is easy to check that the material points present a smooth remodeling. In fact, $\Omega_{X} \left( \mathbb{R} \right)$ is a Lie subgroupoid of $\Phi \left( \mathcal{V} \right)$ if, and only if, $\mathbb{R}$ can be covered by local sections of the anchor of $\Omega_{X} \left( \mathbb{R} \right)$, these sections induce the smooth remodeling process at $X$ (see definition \ref{1.7.2.second}).\\
Roughly speaking, all the particles present a smooth remodeling ($\Omega_{X} \left( \mathbb{R} \right)$ is a Lie subgroupoid of $\Phi \left( \mathcal{V} \right)$) and the variation at different points is also smooth ($\Omega \left( \mathcal{C}\right)$ is a Lie subgroupoid of $\Phi \left( \mathcal{V} \right)$).

\begin{theorem}\label{computationalproposition24124}
Let be a body-time manifold $\mathcal{C}$. $\mathcal{C}$ presents a smooth remodeling process if, and only if, 

\begin{itemize}
    \item[i)] $dim  \left( A \Omega \left( \mathcal{C} \right)^{T}_{\epsilon \left(\left(t,X\right)\right)}\right)  $ is constant respect to $\left(t,X \right)$\\
    \item[ii)] $dim  \left( A \Omega_{X} \left( \mathbb{R} \right)^{\sharp}_{\left(t,X\right)}\right)= 1 $, for all $\left( t , X \right) \in \mathcal{C}$
\end{itemize}
Here, $ A \Omega \left( \mathcal{C} \right)^{T}_{\epsilon \left(\left(t,X\right)\right)}$ (resp. $A \Omega_{X} \left( \mathbb{R} \right)^{\sharp}_{t}$) is the fibre of $ A \Omega \left( \mathcal{C} \right)^{T}$ (resp. $A \Omega_{X} \left( \mathbb{R} \right)^{\sharp}$) at $\epsilon \left( \left(t,X\right)\right)$ (resp. $t$).

\end{theorem}
\noindent{To prove this theorem we will need an auxiliary lemma.}

\begin{lemma}\label{lemmauxiliar213423}
Let $M$ be a manifold and a path-connected subset $X$ of $M$. Consider a regular foliation $\mathcal{F}$ of $M$ such that
\begin{itemize}
    \item[i)] $X$ is union of leaves of $\mathcal{F}$.
    \item[ii)] $X$ is not a leaf of $\mathcal{F}$.
\end{itemize}
Then, there exists a strictly coarser (singular) foliation of $M$ satisfying $i)$.
\begin{proof}
Assume that $X$ is not a leaf of $\mathcal{F}$. Let be a foliation $\varphi = \left(y^{1}, \hdots , y^{n}\right)$ in a neighborhood $U$ of $x\in M$,
\begin{equation}\label{localexpresionUepsilonq434}
U:= \{ - \epsilon < y^{1} < \epsilon , \hdots , - \epsilon < y^{n} < \epsilon \},
\end{equation}
such that the $k-$dimensional disk $\{ y^{k+1}= \hdots = y^{n} = 0\}$ coincides with the path-connected component of the intersection of $\mathcal{F}\left(x\right)$ with $U$ which contains $x$, and each $k-$dimensional disk $\{ y^{k+1} = c_{k+1} , \hdots y^{n} = c_{n} \}$, where $c_{k+1}, \hdots , c_{n}$ are constants, coincides with the path-connected component of the intersection of some $\mathcal{F}\left(y\right)$ with $ U $. We may shrink $\epsilon$ enough to get that $ U   \cap X$ is path-connected.\\
Let be a point $y$ in $ U \cap X$ which is not contained in $\mathcal{F}\left(x\right)$ (i.e., $\mathcal{F}\left(y\right) \neq \mathcal{F}\left(x\right)$). Then, there exists a differentiable path $\alpha : I \rightarrow   U   \cap X$, with $I=\left[0,1\right]$, such that
$$ \alpha \left( 0 \right) = x , \ \ \ \ \ \ \ \ \ \alpha \left( 1 \right) = y.$$
Then, we will consider 
\begin{equation}
    \mathcal{C} := \{ z \in M \ : \ \mathcal{F} \left( z \right)  \cap  \overline{\alpha} \left( 0,1 \right) \neq \emptyset \}
\end{equation}
In other words, $\mathcal{C}$ is the union of all the leaves in such a way that $\alpha$ cuts to all leaves.\\
So, consider the path $\overline{\alpha}: I \rightarrow \mathcal{U}$, $\mathcal{U}= \varphi^{-1}\left(  U  \cap X \right)$, given by
$$ \overline{\alpha}= \varphi^{-1}\circ \alpha.$$
Then, by using the local expression \ref{localexpresionUepsilonq434},
\begin{small}

\begin{equation}
    \mathcal{C} \cap U : \{ - \epsilon < y^{1} < \epsilon , \hdots , - \epsilon < y^{k} < \epsilon , y^{k+1}= \alpha^{k+1}\left( t \right) ,  \dots ,  y^{k+1}= \alpha^{n}\left( t \right) \}_{t \in \left( 0,1\right)},
\end{equation}
\end{small}
where $\alpha^{i}$ are the coordinates of $\alpha$ respect to $\varphi$. Using the rank theorem we may transform $\varphi$ to get that
\begin{small}

\begin{equation}\label{321singularfoliat}
    \mathcal{C} \cap U : \{ - \epsilon < y^{1} < \epsilon , \hdots , - \epsilon < y^{k} < \epsilon , y^{k+1}= t , 0, \dots ,  0) \}_{t \in \left( 0,1\right)},
\end{equation}
\end{small}
Consider the foliation $\mathcal{G}$ of $M$ such that
\begin{itemize}
    \item $\mathcal{G}\left(z \right) = \mathcal{F}\left(z \right)$ for each $z \notin \mathcal{C}$.
    \item $\mathcal{G}\left(z \right) = \mathcal{C}$ for each $z \in \mathcal{C} $.
\end{itemize}
Obviously, $\mathcal{G}$ is a strictly coarser division of $M$ and satisfies $i)$. Furthermore, it is an easy exercise to prove that $\mathcal{G}$ is a singular foliation (see Eq. (\ref{321singularfoliat})).
\end{proof}

\end{lemma}
\noindent{By separating within path-connected component we may prove the following result.}
\begin{lemma}\label{lemmauxiliar213423second4453}
Let $M$ be a manifold and a subset $X$ of $M$. Consider a regular foliation $\mathcal{F}$ of $M$ such that
\begin{itemize}
    \item[i)] $X$ is union of leaves of $\mathcal{F}$.
    \item[ii)] There is at least a path-connected component of $X$ which is not a leaf of $\mathcal{F}$.
\end{itemize}
Then, there exists a strictly coarser (singular) foliation of $M$ satisfying $i)$.
\end{lemma}
Thus, roughly speaking, for each manifold $M$ and any subset $X$ which is not a submanifold of $M$, the maximal foliation satisfying $i)$ is necessarily singular.
\begin{proof}[Proof of Proposition  \ref{computationalproposition24124}]

Notice that, condition $ii)$ is equivalent to that all the $X-$material groupoids $ \Omega_{X} \left( \mathbb{R} \right)$ are transitive Lie subgroupoids of $\Phi \left( \mathcal{V}  \right)$. So, we only have to deal with condition $i)$, i.e.,
$$dim  \left( A \Omega \left( \mathcal{C} \right)^{T}_{\epsilon \left(\left(t,X\right)\right)}\right),  $$
is constant respect to $\left(t,X \right)$. Then, the material foliation $\overline{\mathcal{F}}$ is regular.\\
On the one hand, $\overline{\mathcal{F}}$ is maximal foliation whose leaves are contained in the $\overline{\beta}-$fibres (see corollary \ref{10.33}). Then, taking into account Lemma \ref{lemmauxiliar213423second4453}, the path-connected components of the $\overline{\beta}-$fibres of $\Omega \left( \mathcal{C}\right)$ have to be leaves of the foliation. Finally, the local charts of the transitive Lie subgroupoids $\Omega \left(\overline{\mathcal{F}} \left( x \right) \right)$ (see proof of Corollary \ref{4.4.second2324.uniform2445}) defines a structure on $\Omega \left( \mathcal{C} \right)$ of Lie subgroupoid of $\Phi \left( \mathcal{V} \right)$.

\end{proof}
Hence, Proposition  \ref{computationalproposition24124} provides us a computational way of dealing with the smooth remodeling processes. In other words, by Eq. (\ref{Eqmaterialgroupoid12}) and Eq. (\ref{Xmatgroup323456}), $\mathcal{C}$ is presenting a smooth remodeling if, and only if, the space of solutions of the equation,
\begin{equation}\label{otramas2324}
    \lambda \dfrac{\partial W}{\partial  t}  +  \Theta^{i}\dfrac{\partial W}{\partial  x^{i} } + \Theta^{i}_{j}\dfrac{\partial W}{\partial  y^{j}_{i} } = 0.
\end{equation}
where $\lambda$, $\Theta^{i}$ and $\Theta^{i}_{j}$ are functions depending on $t$ and $X$, has constant dimension and there exists a solution of,
\begin{equation}\label{estaesotra123}
    \lambda\dfrac{\partial W_{X}}{\partial  t } + \Theta^{i}_{j}\dfrac{\partial   W_{X}}{\partial  y^{j}_{i} } = 0
\end{equation}
with $\lambda \neq 0$. Notice that, if it were satisfied Eq. (\ref{estaesotra123}), the space of solutions of Eq. (\ref{otramas2324}) has, at least, dimension $1$.

\part{Aging}

\begin{definition}\label{1.17.2}
\rm
Let $\mathcal{C}$ be a body-time manifold. A material particle $X \in \mathcal{B}$ is presenting a \textit{aging} when it is not presenting a remodeling, i.e., not all the instants are connected by a material isomorphism. $\mathcal{C}$ is a \textit{process of aging} if it is not a process of remodeling.
\end{definition}
Clearly, if the material response is not preserved along the time via material isomorphism, the constitutive properties are changing with the time. Altough it is something natural, there is not a proper definition of \textit{smooth aging}. The presentation of this definition is other of the contributions of this paper (Definition \ref{smoothaging342345}).

\begin{proposition}\label{auxprop4342}
Let $\mathcal{C}$ be a body-time manifold. A material particle $X \in \mathcal{B}$ is presenting an aging if, and only if, the $X-$material groupoid $\Omega_{X}\left(  \mathbb{R}\right)$ is not transitive. $\mathcal{C}$ is presenting an aging if, and only if, for some material point $X$, the $X-$material groupoid $\Omega_{X}\left(  \mathbb{R}\right)$ is not transitive.
\end{proposition}

\begin{corollary}
Let $\mathcal{C}$ be a body-time manifold with some state uniform. $\mathcal{C}$ is presenting an aging if, and only if, the material groupoid $\Omega \left( \mathcal{C}\right)$ is not transitive.
\end{corollary}
\noindent{Then, we are ready to present a definition of \textit{smooth aging}.}

Analogously to smooth remodeling, to define \textit{smooth aging} of the global body-time manifold as the smooth aging of all the material particles is not enough. We need also to impose smoothness on the variation along the material particles.
\begin{definition}\label{smoothaging342345}
\rm
Let be a body-time manifold $\mathcal{C}$. $\mathcal{C}$ is presenting a \textit{smooth aging} if $\Omega \left( \mathcal{C}\right)$ is a Lie subgroupoid of $\Phi \left( \mathcal{V} \right)$ and, there is a particle $X$ such that the $X-$material groupoid $\Omega_{X} \left( \mathbb{R} \right)$ is a not transitive Lie subgroupoid of $\Phi \left( \mathcal{V} \right)$.

\end{definition}
In other words, $\mathcal{C}$ is presenting a smooth aging if the variation of the body is ``\textit{smooth}'' ($\Omega \left( \mathcal{C}\right)$ is a Lie subgroupoid of $\Phi \left( \mathcal{V} \right)$) and it is not a smooth remodeling.\\
Observe that, taking into account Proposition \ref{auxiliarprop34re}, if $\Omega \left( \mathcal{C}\right)$ is a Lie subgroupoid of $\Phi \left( \mathcal{V} \right)$, then for all particle $X$, the $X-$material groupoid $\Omega_{X} \left( \mathbb{R} \right)$ is a Lie subgroupoid of $\Phi \left( \mathcal{V} \right)$. So, the unique imposition is given over the lack of transitivity of a $X-$material groupoids.\\\\

\noindent{Consider $\Omega \left( \mathcal{C} \right)$, the material groupoid associated to the body-time manifold $\mathcal{C}$. Then, we have available the material distribution $A \Omega \left( \mathcal{C} \right)$, body-material distribution $A \Omega \left( \mathcal{C} \right)^{\sharp}$ and the uniform-material distribution $A \Omega \left( \mathcal{C} \right)^{B}$ and their associated foliations, the material foliation $\overline{\mathcal{F}}$, body-material foliation $\mathcal{F}$ and uniform-material foliation $\mathcal{G}$, respectively.}\\

\begin{proposition}\label{computationalproposition24124aging}
Let be a body-time manifold $\mathcal{C}$. $\mathcal{C}$ presents a smooth aging process if, and only if, 

\begin{itemize}
    \item[i)] $dim  \left( A \Omega \left( \mathcal{C} \right)^{T}_{\epsilon \left(\left(t,X\right)\right)}\right)  $ is constant respect to $\left(t,X \right)$\\
    \item[ii)] For some $X$, $dim  \left( A \Omega_{X} \left( \mathbb{R} \right)^{\sharp}_{t} \right)= 0 $, for some $t$.
\end{itemize}
Here, $ A \Omega \left( \mathcal{C} \right)^{T}_{\epsilon \left(\left(t,X\right)\right)}$ (resp. $A \Omega_{X} \left( \mathbb{R} \right)^{\sharp}_{t}$) is the fibre of $ A \Omega \left( \mathcal{C} \right)^{T}$ (resp. $A \Omega_{X} \left( \mathbb{R} \right)^{\sharp}$) at $\epsilon \left( \left(t,X\right)\right)$ (resp. $t$).

\begin{proof}
Notice that, if $\Omega \left( \mathcal{C}\right)$ is a Lie subgroupoid of $\Phi \left( \mathcal{V} \right)$, then, by Proposition \ref{auxiliarprop34re}, all the $X-$material groupoids are Lie subgroupoids of $\Phi \left( \mathcal{V} \right)$. Then, $dim  \left( A \Omega_{X} \left( \mathbb{R} \right)^{\sharp}_{t}\right)$ does not depend on time at any instant $t$. Therefore, $\Omega_{X} \left( \mathbb{R} \right)$ is not transitive if, and only if, $dim  \left( A \Omega_{X} \left( \mathbb{R} \right)^{\sharp}_{t}\right)= 0 $, for all $t$. Then, the proof is analogous to theorem \ref{computationalproposition24124}.
\end{proof}
\end{proposition}

In this way, again, we present a result characterizing smooth aging of the evolution material which gives a computational way of testing it.\\\\

\begin{definition}\label{1.17.5}
\rm
Let $\mathcal{C}$ be a body-time manifold. The body $\mathcal{B}$ is said to be undergone a \textit{uniform aging} if for each $t \in \mathbb{R}$ all the points $\left( t , X \right) \in \mathcal{C}$ are isomorphic and it is not presenting a uniform remodeling.
\end{definition}
Intuitively, in a process of uniform aging the material properties change equally at all the points.

\begin{proposition}

Let $\mathcal{C}$ be a body-time manifold. The body $\mathcal{B}$ presents uniform aging if, and only if, for all $t$, the $t-$material groupoid $\Omega_{t} \left( \mathcal{B} \right)$ is transitive and the material groupoid $\Omega \left( \mathcal{C}\right)$ is not transitive.
\end{proposition}
So, a process of aging is uniform if all the states of the body are uniform but the intrinsic properties of the body vary along the time.
\begin{proposition}
Let be a body-time manifold $\mathcal{C}$. $\mathcal{C}$ presents a smooth uniform aging process if, and only if, 

\begin{itemize}
    \item[i)] $dim  \left( A \Omega \left( \mathcal{C} \right)^{T}_{\epsilon \left(\left(t,X\right)\right)}\right)  $ is constant respect to $\left(t,X \right)$\\
    \item[ii)] For some $X$ and $t$, $dim  \left( A \Omega_{X} \left( \mathbb{R} \right)^{\sharp}_{t} \right)= 0 $.
    \item[iii)] For all $t$ and some $X$, $dim  \left( A \Omega_{t} \left( \mathcal{C}\right)^{\sharp}_{X} \right)= 3 $.
\end{itemize}
Notice that, by proposition \ref{auxiliarprop34re}, $dim  \left( A \Omega_{X} \left( \mathbb{R} \right)^{\sharp}_{t} \right)$ and $dim  \left( A \Omega_{t} \left( \mathcal{C}\right)^{\sharp}_{X} \right)$ are constant on $t$ and $X$, respectively.\\\\

\end{proposition}

\section*{Acknowledgments}
M. de Leon and V. M. Jiménez acknowledge the partial finantial support from MICINN Grant PID2019-106715GB-C21 and the ICMAT Severo Ochoa project CEX2019-000904-S.

\bibliographystyle{plain}

\bibliography{Library}

\end{document}